\documentclass[journal,onecolumn,draftcls,12pt,twoside]{IEEEtran}
\normalsize
\usepackage{color,graphicx,amsmath,amssymb,amsthm,mathrsfs,cite}
\usepackage{multirow}

\usepackage{tikz}
\usepackage{multicol}
\usepackage{lipsum,adjustbox}
\usepackage{flushend}
\usepackage{cuted}
\usepackage{enumerate}
\usepackage{cleveref}
\usepackage{float}
\usepackage{makecell}
\usepackage{cool}
\usepackage{caption,subcaption}
\usetikzlibrary{decorations.markings}
\usetikzlibrary{arrows,snakes}
\usetikzlibrary{intersections}
\usetikzlibrary{plotmarks}
\usepackage{pgfplots}
\usepackage{algorithm}
\usepackage{algorithmic}
\usepackage{caption}
\usepackage{graphics}
\usepackage{etoolbox}
\AfterEndEnvironment{strip}{\leavevmode}

\begin{document}
\tikzstyle{bnode}=[circle,draw,fill=black!50,minimum size=1mm]
\tikzstyle{cnode}=[circle,draw]
\tikzstyle{cgnode}=[circle,draw]
\tikzstyle{cnode}=[circle,draw]
\tikzstyle{crnode}=[circle,draw]
\tikzstyle{conode}=[rectangle,draw,minimum size=1cm]
\tikzstyle{cenode}=[rectangle,draw,minimum size=0.67cm]
\tikzstyle{ccnode}=[rectangle,draw,minimum size=1.7cm]
\tikzstyle{cpnode}=[circle,draw]
\tikzstyle{cmnode}=[circle,draw,minimum size=2cm]
\tikzstyle{rnode}=[rectangle,draw,outer sep=0pt]
\tikzstyle{prnode}=[rectangle,rounded corners,fill=blue!50,text width=4.5em,text centered,outer sep=0pt]
\tikzstyle{prnodebig}=[rectangle,rounded corners,fill=blue!50,text width=7em,text centered,outer sep=0pt]
\tikzstyle{prnodesimple}=[rectangle,draw,text width=4.5em,text centered,outer sep=0pt]
\tikzstyle{bigsnake}=[fill=green!50,snake=snake,segment amplitude=4mm, segment length=4mm, line after snake=1mm]
\tikzstyle{smallsnake}=[snake=snake,segment amplitude=0.7mm, segment length=4mm, line after snake=1mm]
\tikzstyle{anode}=[arc,draw]
\tikzset{middlearrow/.style={
        decoration={markings,
            mark= at position 0.5 with {\arrow{#1}} ,
        },
        postaction={decorate}
    }
}
\newtheorem{theorem}{Theorem}
\newtheorem{lemma}{Lemma}
\newtheorem{example}{Example}
\newtheorem{remark}{Remark}
\newcommand{\C}{\mathcal{C}}
\newcommand{\ch}{\text{ch}}
\newcommand\Tstrut{\rule{0pt}{2.6ex}}       % "top" strut
\newcommand\Bstrut{\rule[-0.9ex]{0pt}{0pt}}
\newcommand{\TBstrut}{\Tstrut\Bstrut}
\newcommand{\specialcell}[2][c]{%
  \begin{tabular}[#1]{@{}c@{}}#2\end{tabular}}
\def\e{\epsilon}
\def\th{\text{th}}

\title{Protograph LDPC Codes with Block Thresholds: Extension to  Degree-1  and Generalized Nodes}

\author{\IEEEauthorblockN{Asit Kumar Pradhan and Andrew Thangaraj}\\
\IEEEauthorblockA{Department of Electrical  Engineering\\
Indian Institute of Technology Madras,
Chennai 600036, India\\
Email: asit.pradhan,andrew@ee.iitm.ac.in}}

\maketitle
\begin{abstract}
\textbf{
Protograph low-density-parity-check (LDPC) codes  are considered to design near-capacity low-rate codes over the binary erasure channel (BEC) and binary additive white Gaussian noise (BIAWGN) channel. For protographs with degree-one variable nodes and doubly-generalized LDPC (DGLDPC) codes,  conditions are derived to ensure equality of bit-error threshold and block-error threshold. Using this condition, low-rate codes with block-error threshold close to capacity are designed and shown to have better error rate performance than other existing codes.} 
\end{abstract}

\section{Introduction}

Low-density-parity-check (LDPC) codes\cite{gallagerthesis}, introduced by Gallager in 1960s, became popular in 1990s, because of their excellent performance under iterative message-passing decoding\cite{Mackay}. 
Several applications including security applications like wiretap coding \cite{Arun} need code sequences with provable block-error threshold, i.e for any channel parameter below bit-error threshold, block-error rate should provably approach zero as blocklength goes to infinity. Two early efforts in block-error threshold for LDPC codes  include consideration of ML decoding \cite{gallagerthesis}  and a stopping set distribution based analysis \cite{orlitsky}.  In \cite{Lentermaier}, authors have shown that block-error and bit-error thresholds are the same under message-passing decoder for codes having minimum variable node degree greater than two.  In our earlier work\cite{Pradhan}, we have extended the block-error condition in \cite{Lentermaier} allowing degree-two nodes in protograph LDPC ensembles under the condition that the degree-two subgraph of the protograph is a tree. For BEC, the best reported rate-$1/2$ degree distribution with minimum variable node degree three has threshold $0.4619$\cite{Arun}, which is away from capacity. In\cite{Pradhan}, we have also designed high rate codes (rate $ \geq 1/2$) with block-error threshold close to capacity using differential evolution. However, computer search shows that low-rate  codes (rate $\leq 1/3$) satisfying block-error threshold condition derived in \cite{Pradhan} have bit-error threshold away from capacity. For example, a $7\times 8$ optimized protograph defining a rate-$1/8$ code has a gap of $0.3$ between block-error threshold and capacity over the Binary Erasure Channel (BEC). Low-rate codes with bit-error threshold close to capacity have a large fraction of degree-one bit nodes\cite{1523619}\cite{Richardson2004}, which are not allowed by the block-error threshold condition in \cite{Pradhan}. In this work,  we extend the block-error threshold condition in \cite{Pradhan} to allow degree-one bit nodes, which  play an important role in designing low-rate codes with block-error threshold close to capacity. Protographs in 5G standard\cite{5G} and protograph-based raptor like codes (PBRL)\cite{pbrl} satisfy the block-error threshold condition derived in this paper, while AR4A protographs\cite{1523619}  do not satisfy the block-error threshold condition, which has been validated by simulation results in Section \ref{sec:optim-dgldpc-prot}. Using the new block-error threshold condition, we have designed low-rate codes with block-error threshold close to capacity. For example, we have designed a rate-$1/8$ protograph LDPC code for BEC  with a block-error threshold of $0.866$ (gap of $0.009$ from capacity). For the binary additive white Gaussian noise (BIAWGN) channel, we have designed a code of rate-$1/3$ and  blocklength-$64800$ which has better BER/FER performance than the rate-$1/3$ protograph based raptor like (PBRL) code in\cite{pbrl}.
  
  We also extend the block-error threshold condition to protograph Doubly Generalized LDPC (DGLDPC) codes, introduced in \cite{wangDGLDPC} and studied in \cite{4787610,4939225,5437431,5695100}.  Block-error condition for protograph DGLDPC ensembles allows cycles even if degree of all the variable nodes in the cycle is two and enables better optimization of codes.
% Using protograph DGLDPC code, we also design rate-compatible codes, needed in applications like wireless communication. Traditionally, puncturing or extension is used for designing  rate-compatible LDPC codes \cite{Klinc,Anastasopoulos,Nguyen}. In standards, multiple protographs are used for different rates.  In a GLDPC protograph, we introduce rate-compatibility  by changing component codes at check nodes without altering any connections of the protograph. Using differential evolution, we simultaneously optimize thresholds corresponding to different component codes  at check nodes and obtain rate-compatible codes with reasonable gap to capacity, over a wide range of rates.

Rest of the paper is organized as follows. Section \ref{sec:defin-prel} introduces definition and notation for protograph DGLDPC codes and describes density evolution over BEC and BIAWGN channel.  Section \ref{sec:prop-block-error} derives conditions on protograph and component codes under which block-error threshold equals bit-error threshold for large-girth ensembles. Optimization of protograph DGLDPC code is described in Section \ref{sec:optim-dgldpc-prot}.  %Design  of rate-compatible codes is described in Section \ref{sec:rate adaptation}. 

%Our contributions are the following: (1) exact density evolution over the erasure channel,
%(2) conditions on protograph and component codes under which
%block-error threshold equals density evolution threshold for
%large-girth ensembles, (3)
%stability of the density evolution recursion, (4) extension to
%binary-input symmetric channel using Bhattacharya parameter, and
%(5) protograph optimization for near-capacity performance. Many of the
%results have been derived for the first time for protograph DGLDPC
%codes, and are based on similar results for protograph LDPC codes in
%\cite{Pradhan}. One of the advantages of generalized and doubly
%generalized codes is the flexibility they offer in terms of range of possible
%rates, particularly low rates with performance close to capacity. We
%demonstrate this flexibility through some optimized examples.

% This article is organized as follows. In Section \ref{sec:defin-prel}, we introduce
% protograph DGLDPC codes and exact density evolution over the binary
% erasure channel. In Section \ref{sec:prop-block-error}, we discuss the
% double-exponential fall property of message-error probability in
% density evolution, which is crucial for the existence of block-error
% threshold in large-girth ensembles. Stability and extension to 
% binary-input symmetric channels are also discussed in Section
% \ref{sec:prop-block-error}. We present some optimized protograph
% GLDPC and DGLDPC codes in Section \ref{sec:optim-dgldpc-prot}.

\section{Definitions and Preliminaries}  
\label{sec:defin-prel}
%Low-density parity-check (LDPC)  codes and their Tanner graph
%representation are well-known. The left-side nodes of the bipartite
%Tanner graph represent bits of the codeword, and the right-side nodes
%represent parity checks. If a check node is connected to $d$ bit 
%nodes, then those $d$ bits belong to the $(d,d-1)$ single parity check
%(SPC) code. If a bit node is connected to $d$ check nodes, then that
%bit, repeated $d$ times, participates in those $d$ checks. Hence, a
%bit node enforces a repetition code, while a check node enforces the
%SPC code. The design rate of the LDPC code is given by $1-\frac{M}{L}$,
%where $M$ and $L$ are the number of check and bit nodes, respectively. 
A general block-error threshold condition will be derived for doubly-generalized low density parity check (DGLDPC) codes. We will first define these codes formally and introduce notation for protograph  DGLDPC codes.

\subsection{Protograph DGLDPC codes}
Protograph LDPC codes are defined by Tanner graphs that are created
from a small base graph, called protograph, by a copy-permute
operation. Protograph DGLDPC codes are defined in a similar way by
allowing the variable and check nodes of the protograph to enforce
arbitrary linear codes as component codes. 

Fig. \ref{fig:protex} is an example of a protograph that expands
to a DGLDPC code.
\begin{figure}[htb]
\centering
\begin{tikzpicture}[every node/.style={scale=0.8}]
\begin{scope}[node distance=2cm,>=angle 90,semithick]
\node[cnode] (v1) {$v_1$};
\node[crnode] (v2)[below of=v1] {$v_2$};
\node[cgnode] (v3)[below of=v2] {$v_3$};
\node[cpnode,double] (v4)[below of=v3] {$v_4$};
\node[conode,double] (c1)[right of=v2,xshift=2cm,yshift=1cm] {$c_1$}
  edge[black] node[very near end,above]{$e_{v_1,1}$} node[very near start,above]{$e_{c_1,1}$}(v1)
  edge[blue] node[ near start,above]{$e_{c_1,2}$} (v2);
\draw[gray] (v4.65) -- node[very near start,left]{$e_{v_4,1}$}(c1.220);
\draw[cyan] (v4.35) -- node[very near start,right]{$e_{v_4,2}$}(c1.245);
\node[conode] (c2)[right of=v3,xshift=2cm,yshift=-1cm] {$c_2$}
  edge[magenta] node[very near start,right]{$e_{c_2,1}$} (v1)
  edge[violet] node[near start,above]{}(v4);
\draw[brown] (v3) --node[very near start,above]{} (c1);
\draw[orange] (v3.335) --node[very near start,above]{} (c2.180);
\draw[red] (v2) --node[very near end,left]{$e_{c_2,2}$}(c2);
\draw[green] (v4.-10) --node[very near start,above]{}(c2.207);
\draw[purple] (v4.-30) --node[very near start,above]{}(c2.215);
\node (b1v4) [left of =v4,yshift=0.13cm,xshift=0.7cm]{};
\node(b2v4) [left of =v4,yshift=-0.13cm,xshift=0.7cm]{};
\draw[black](b1v4) --node[]{} (v4.162);
\draw[black](b2v4) --node[]{} (v4.1999);
\node (b1v3) [left of =v3,xshift=0.7cm]{};
\draw[black](b1v3) --node[]{} (v3);

\node (b1v2) [left of =v2,xshift=0.7cm]{};
\draw[black](b1v2) --node[]{} (v2);

\node (b1v1) [left of =v1,xshift=0.7cm]{};
\draw[black](b1v1) --node[]{} (v1);

\end{scope}    
\end{tikzpicture}
\caption{Example of a protograph for DGLDPC codes. Double line for a
  node indicates that its component code is generalized.}
\label{fig:protex}
\end{figure}
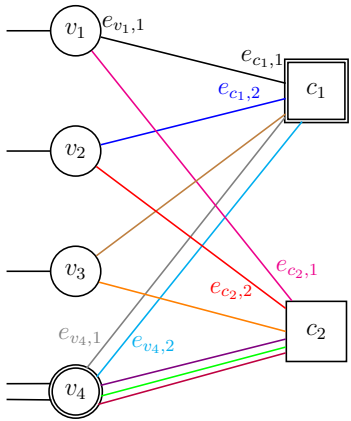
The variable node $v_4$ 
and the check node $c_1$ have a (5,2) linear code as component
code. All other check nodes and variable nodes enforce single parity check code and repetition code, respectively. 

A protograph is denoted as $G=(V\cup
C,E)$, where $V$ and $C$ are the set of variable and
check nodes, respectively, and $E$ is the set of undirected edges. Multiple
parallel edges are allowed between a variable node and a check
node in a protograph. 
% The nodes and edges in the protograph are
% ordered, and the $i$-th variable node and check node  in the
% protograph are denoted, respectively, $v_i$, $c_i$. The degree of
% $v_i$ and $c_i$ are denoted $d_{vi}$ and $d_{ci}$, respectively, and
% their component codes are $(d_{vi},k_{vi})$ and $(d_{ci},k_{ci})$
% linear codes.
Let $d_{v_i}$ and $d_{c_j}$ denote the degree of variable node $v_i$ and check node $c_j$, respectively. 
The edges connected to a variable node $v_i$ or a check node $c_j$ are denoted by $e_{v_i,m}$ and $e_{c_j,n}$, respectively, where $m\in\{1,2\ldots d_{v_i}\}$ and $n\in\{1,2\ldots d_{c_j}\}$.
If $v_i$ is connected to $c_j$, then $e_{v_i,m} = e_{c_j,n}$ for some
$m$, $n$. The variable and check nodes connected to edge $e$ are denoted by $v_e$ and $c_e$, respectively.

The lifted or expanded graph is obtained by the copy-permute
operation\cite{Thrope}, and is specified by the number of copies
and a permutation for each edge type. First, a given protograph is copied $T$ times. Variable node $v_i$, check node $c_j$, and edge $e_{v_i,m}$ of $t$-th copy of protograph are denoted by $(v_i,t), (c_j,t)$ and $(e_{v_i,m},t)$, respectively. For each edge $e_{v_i,m}$ of protograph, we assign a permutation $\pi_{i,m}$ of the set $\{1,2\ldots T\}$.
 If $v_i$ is connected to $c_j$ by $e_{v_{i,m}}$ in the protograph, then after permutation operation, 
the edge $(e_{v_i,m},t)$ connects the variable node $(v_i,t)$ to check node $(c_j,\pi_{i,m}(t))$. Copies of edge $e_{v_i,m}$, variable node $v_i$ and
check node $c_j$ in the lifted graph are said to be of type
$e^v_{i,m}$, $v_i$ and $c_j$, respectively. In the copy
operation,  variable node  $(v_i,t)$ and check node $(c_j,t)$ will,  respectively, have the same component code as variable node $v_i$ and check node $c_j$ of the protograph.  The design rate of the lifted graph is the same as that of
the protograph. 
For the lifted graph,  component code at each node and edge types in the computation graph is completely determined by the protograph $G$. Girth of a graph is defined as the least length of a cycle in the graph. On a girth-$g$ lifted graph,  protograph density evolution analysis is accurate up to iteration $\frac{g}{2}-1$.
%\subsection{Message passing decoder}
%\begin{figure}[htb]
%\centering
%\begin{tikzpicture}[every node/.style={scale=0.8}]
%\begin{scope}[node distance=2cm,>=angle 90,semithick]
%\node[cmnode,double] (v1) {$v_i$};
%\node[ccnode,double] (c1)[right of=v1,xshift=2cm,yshift=0cm] {$c_j$};
%\draw[black][->] (v1.330) --node[very near start,below]{$r^{(t)}_{v_i,d_{v_i}}$}+(330:1cm);
%\draw[black][dashed,<-] (v1.340) --node[very near end,above]{$s^{(t)}_{v_i,d_{v_i}}$} +(332:1cm);
%\node[right of=v1,xshift=-0.7cm,yshift=0.1cm]{$\vdots$};
%\draw[black][dashed,<-] (v1.20) --node[very near end,below]{$s^{(t)}_{v_i,1}$} +(20:1cm);
%\draw[black][->] (v1.30) --node[very near start,above]{$r^{(t)}_{v_i,1}$}+(22:1cm);
%
%\node (b1v1) [left of =v1,xshift=-0.5cm,yshift=0.3cm]{$y_{v_i,1}$};
%\draw[black](b1v1) --node[]{} (v1.163);
%\node[left of=v1,xshift=0.5cm,yshift=0.1cm]{$\vdots$};
%\node (b1v2) [left of =v1,yshift=-0.3cm,xshift=-0.5cm]{$y_{v_i,kv_i}$};
%\draw[black](b1v2) --node[]{} (v1.197);
%\draw[black][->] (c1.330) --node[very near end,below]{$s^{(t)}_{c_j,d_{c_j}}$}+(330:1cm);
%\draw[black][dashed,<-] (c1.340) --node[very near end,above]{$r^{(t)}_{c_j,d_{c_j}}$} +(332:1cm);
%\draw[black][dashed,<-] (c1.20) --node[very near end,below]{$r^{(t)}_{c_j,1}$} +(20:1cm);
%\draw[black][->] (c1.30) --node[very near end,above]{$s^{(t)}_{c_j,1}$}+(22:1cm);
%\node[right of=c1,xshift=-0.8cm,yshift=0.1cm]{$\vdots$};
%\end{scope}    
%\end{tikzpicture}
%\caption{Message passing decoder at genralized check and variable node}
%\label{fig:MAP decoding}
%\end{figure}

Iterative message-passing decoding of DGLDPC codes is a generalized version of iterative decoding used for standard LDPC codes. At check node $c_j$, extrinsic Maximum \emph{A Posteriori} (MAP) processing is performed using the enforced $(d_{c_j},k_{c_j})$  component code.  Given input log-likelihood ratio (LLR) at a check node $c_j$, computation of extrinsic LLR is described below. Codeword associated with component code at check node $c_j$ are denoted  by $\mathbf{z}_j=(z_{j,1},z_{j,2},\cdots z_{j,d_{c_j}}).$ Let $\mathbf{r}_j=(r_{j,1},r_{j,2}, \cdots ,r_{j,d_{c_j}})$ be the a input LLR at check node $c_j$. Then, extrinsic output LLR, denoted by $\mathbf{s}=(s_{j,1},s_{j,2}, \cdots s_{j,d_{c_j}})$, can be computed as follows.
\begin{align*}
s_{j,m}&=\log \frac{P(z_{j,m}=0|\mathbf{r}_{j\setminus m})}{P(z_{j,m}=1|(\mathbf{r}_{j\setminus m})},\\
&=\log \frac{\displaystyle\sum_{\mathbf{z_j}:z_{j,m}=0}P((\mathbf{r}_{j\setminus m}|\mathbf{z}_{j\setminus m})}{\displaystyle\sum_{\mathbf{z_j}:z_{j,m}=1}P(\mathbf{r}_{j\setminus m}|\mathbf{z}_{j\setminus m})},\\
\end{align*}
where $\mathbf{z_{j\setminus m}}=(z_{j,1}, \cdots z_{j,m-1},z_{j,m+1} \cdots z_{j,d_{c_j}})$ and $\mathbf{r_{j\setminus m}}=(r_{j,1}, \cdots r_{j,m-1},r_{j,m+1} \cdots r_{j,d_{c_j}}).$
At variable node $v_i$ enforcing a $(d_{v_i},k_{v_i})$
component code with generator matrix $\mathcal{G}$, channel information for $k_{v_i}$
bits is combined with incoming messages from check nodes in the previous iteration by extrinsic MAP processing on the extended component code with generator matrix $[\mathcal{G} |I_{k_{v_i}}]$, where $I_{k_{v_i}}$ is
the $k_{v_i}\times k_{v_i}$ identity matrix. Computation of extrinsic LLR at variable node is similar to check node.
\subsection{Density evolution over BEC}
\label{sec:density_evolution_BEC}
%\subsubsection{Binary Erasure channel}
Consider iterative message passing decoding on the lifted 
Tanner graph $G'$ derived from a DGLDPC protograph $G=(V\cup C,E)$
after transmission over a
BEC with erasure probability $\epsilon$. In iteration $t$, let
$x^t_{v_i,m}$ denote the probability that  an edge of type $e_{v_i,m}$ carries an
erasure from variable node to check node. Since the lifted graph has $|E|$ edge types, density evolution
is a vector recursion that proceeds by computing $x^{t+1}_{v_i,m}$ for, $1 \leq i \leq |V|, 1 \leq m \leq d_{v_i}$, from the vector $\{x^t_{v_i,m}\}$.  
Let $y^t_{c_j,n}$ denote the probability that an edge of type $e_{c_j,n}$ carries an
erasure from check node to variable node in the $t$-th iteration.
 If nodes connected to an edge $e$ are not relevant in some context, then erasure probability on $e$ from variable node to check node  and check node to variable node after $t$ iteration are denoted by $x^t_e$ and $y^t_e$, respectively.
Since MAP processing is done with the component code at check and
variable nodes, the evolution of $x^t$ and $y^t$ will depend on
the extrinsic messages generated by MAP decoders of the component
codes. 
% Let us denote the extrinsic message of the $m$-th bit generated by the MAP decoder of component code at $c_j$ as $f_m^{c_j}$ for $1\le m\le d_{cj}$. 
% As explained earlier, at a variable node, MAP decoding is done over
% the extended code to exploit channel information.
% Denote the extrinsic message of the $n$-th bit generated by the MAP
% decoder of the extended component code at $v_i$ as $f_n^{v_i}$ for
% $1\le n\le d_{vi}$. 
% It is important to observe that $f_m^{v_i}$ and $f_n^{c_j}$ are
% different, in general, for each $m$ and $n$, respectively. So, the order of edges connected 
% to a variable or check node may affect the performance of the
% ensemble. 
For obtaining an explicit expression for the probability of erasure of an
extrinsic message generated by the MAP decoder of a linear code, a
method based on the support weights and information functions of the
linear code is used as described and discussed in \cite{Sharon}. An alternative method based on multi-dimensional input/output transfer functions of the component decoders has been described in \cite{LentmaierBECProtoDE} to obtain expression for the probability of erasure of an extrinsic message. An example of a $(5,2)$ linear code, worked out in \cite{LentmaierBECProtoDE}, is reproduced here and used later in an illustrative example of density evolution.
\begin{example}
Consider the (5,2) linear code with codewords
$\{00000,01011,10101,11110\}$. Let $x_i$, $i=1,2,\ldots,5$, denote the independent
input erasure probabilities of the 5 bits. Let $h_i(x_{\sim i})$,
where $x_{\sim i}=\{x_1,\ldots,x_5\}\setminus \{x_i\}$ is the list of all variables except $x_i$, denote
the output erasure probability of bit $i$. From
\cite{LentmaierBECProtoDE}, $h_i$ can be explicitly written in terms of $x_i$. For example, $h_1$ and $h_3$ are as follows:
  \begin{align}
 %\label{eq:8}
  h_1(x_2,x_3,x_4,x_5)&=x_3x_5+x_2x_3x_4-x_3x_4x_5x_2,\nonumber\\
  h_3(x_1,x_2,x_4,x_5)&=x_1x_5+x_1x_2x_4-x_1x_4x_5x_2.
 \end{align}
\label{ex:EXIT}
\end{example}
To proceed further, we assume that the extrinsic message-erasure
probabilities from MAP processing at $m$-th edge of node $v_i$ and $n$-th edge of node $c_j$ have been
derived, and these are denoted as $h_{v_i,m}(\cdot)$ and
$h_{c_j,n}(\cdot)$, respectively. With this notation, the protograph density
evolution recursion is given by the following equations for $1\le i\le
|V|$, $1\le m\le d_{v_i}$, $1\le j\le |C|$, $1\le n\le d_{c_j}$:
\begin{align}
 \label{eq:5}
 x^0_{v_i,m}&=h_{v_i,m}(\mathbf{1}_{d_{v_i}-1},\epsilon\mathbf{1}_{k_i}),\\
 \label{eq:6}
 y^{t+1}_{c_j,n}&=h_{c_j,n}\left(\mathbf{x}^t_{c_j,\sim n}\right), \\
 \label{eq:7}
 x^{t+1}_{v_i,m}&=h_{v_i,m}\left(\mathbf{y}^t_{v_i,\sim m},\epsilon\mathbf{1}_{k_i} \right).
\end{align}
where $\mathbf{x}^t_{c_j, \sim n}=\{x^t_{c_j,1},\cdots,x^t_{c_j,n-1},x^t_{c_j,n+1},\cdots,x^t_{c_j,d_{c_j}}\},$ \\ $\mathbf{y}^t_{v_i, \sim m}=\{y^t_{v_i,1},\cdots,y^t_{v_i,m-1},y^t_{v_i,m+1},\cdots,y^t_{v_i,d_{v_i}}\}$ and $\mathbf{1}_k$
is the length-$k$ all-ones vector. In the first iteration, shown in \eqref{eq:5}, the probability that an incoming message from
a check node is an erasure is set as 1. Erasure probability from
the channel is set to be $\epsilon$. 
%To proceed further, the output erasure probability iscomputed for the MAP decoder of the extended version of the
%component code at $v_i$ in \eqref{eq:5}. Subsequent iterations use MAP processing
%functions at the variable and check nodes for evolution as shown in
%\eqref{eq:6} and \eqref{eq:7}.
%In the next example, density evolution recursion is illustrated for the protograph in Fig. \ref{fig:protex}.
Example 2 illustrates density evolution recursions for a variable node having a (5,2) linear code as component code.
\begin{example}
\label{ex:density_evolution}
 In Fig. \ref{fig:protex}, let the component codes at the variable node
 $v_4$ and check node $c_1$ be the $(5,2)$ code
 considered in Example \ref{ex:EXIT}. All other variable and check
 nodes enforce repetition codes and SPC codes, respectively. As mentioned
 earlier, at $v_4$ MAP decoding is done over the
 extended version of the $(5,2)$ code with codewords $\{0000000,0101101,1010110,1111011\}$.
% parity check matrix
%  \begin{equation}
%  \label{eq:extend_par_mat}
%   \begin{bmatrix}
%  1 & 0 & 1 & 0 & 0 & 0 & 0 \\
%  0 & 1 & 0 & 1 & 0 & 0 & 0\\
%  1 & 1 & 0 & 0 & 1 & 0 & 0\\
%  1 & 0 & 0 & 0 & 0 & 1 & 0\\
%  0 & 1 & 0 & 0 & 0 & 0 & 1
% \end{bmatrix}.
%  \end{equation}
Before starting recursion, assign $y^0_{v_i,m}=1$ for $1\leq i \leq |V|
$ and $1 \leq m \leq d_{v_i}$. The evolution for a few edges is shown below:
\begin{align*}
& x^{t+1}_{v_1,1}=\epsilon y^{t}_{v_1,2}, \quad x^{t+1}_{v_1,2}=\epsilon y^{t}_{v_1,1},\\
 &x^{t+1}_{v_4,1}=\epsilon y^{t}_{v_4,3}y^{t}_{v_4,5} +\epsilon^2 y^{t}_{v_4,2}y^{t}_{v_4,3}y^{t}_{v_4,4}
-\epsilon^2 y^{t}_{v_4,2}y^{t}_{v_4,3}y^{t}_{v_4,4}y^{t}_{v_4,5},\\ 
%&x_{t+1}(e^v_{4,2})=\epsilon y_{t}(e^v_{4,4})y_{t}(e^v_{4,5})+\epsilon^2 y_{t}(e^v_{4,1})y_{t}(e^v_{4,3})y_{t}(e^v_{4,4})\\
%&\quad \quad \quad \quad \quad-\epsilon^2y_{t}(e^v_{4,1})y_{t}(e^v_{4,3})y_{t+1}(e^v_{4,4})y_{t}(e^v_{4,5}),\\
&y^{t}_{c_1,1}=x^{t}_{c_1,3}x^{t}_{c_1,5} + x^{t}_{c_4,2}x^{t}_{c_1,3}x^{t}_{c_1,4}- x^{t}_{c_1,2}x^{t}_{c_1,3}x^{t}_{c_1,4}x^{t}_{c_1,5}.
% \end{align*}
% \begin{align*}
%&y_{t}(e_{2,1}^c)=1-\left[(1-x_{t}(e_{2,2}^c))(1-x_{t}(e_{2,3}^c))\right.\\
%&\quad \quad \quad \quad \quad\left. (1-x_{t}(e_{2,4}^c))(1-x_{t}(e_{2,5}^c))(1-x_{t}(e_{2,6}^c))\right].\\
\end{align*}
\end{example}

\section{Block-Error Threshold Extensions} 
\label{sec:prop-block-error}
In this section, we generalize block-error threshold conditions  for protograph LDPC codes to protograph with degree-1 variable nodes and to DGLDPC codes. 
The density evolution threshold or bit-error threshold, denoted as $\e_{\text{th}}$, for
the protograph ensemble is defined as the supremum of the set of $\epsilon$ for which erasure probability on each edge tends to zero as iterations tend to
infinity, i.e 
$$\e_{\th}=\sup\{\epsilon:\max\limits_{e\in E}x^t_e\rightarrow0\}.$$
Let us define $\overline{x}^t$ as follows:$$x^t=\sup_{e \in E}x^t(e).$$ Block-error threshold of protograph ensemble is defined as the supremum of the set of $\e$ for which probability of block error, denoted by $P_B$, tends to zero as the number of iterations tends to infinity.
In \cite{Pradhan}, sufficient conditions for block-error threshold being equal to bit-error threshold have been derived using the following two steps:
\begin{itemize}
\item[1.]In first step, it has been shown that $x^t$ falls double exponentially with $t$, i.e. $$\overline{x}^t=\mathcal{O}(\exp(-\beta 2^{\alpha t}))$$ with $\alpha > 0,\beta > 0$, when the degree-two subgraph of the protograph is a tree and $\e \leq \e_{\text{th}}$.
\item[2.]Under the assumption that the girth, denoted by $g$, of the lifted code of blocklength $n$ is $\mathcal{O}(\log n)$ and $t < g/2$, it has been shown that block-error threshold is same as bit-error threshold by upper bounding $P_B$ with $nx^t$ and using double-exponential fall property of $x^t$ as follows:
$$P_B \leq n\overline{x}^t = \mathcal{O}(n\exp(-\beta 2^{\alpha t}))= \mathcal{O}(n\exp(-\beta n^\alpha)).$$  
\end{itemize}
The basic idea in the proof of step one is the following:
when the degree-two subgraph of a protograph is a tree, variable node with degree greater than two is traversed in every $|V|$ (number of variable nodes) iterations of density evolution,  resulting in squaring of $x^t$, which is sufficient to show double exponential fall of $x^t$ as described in \cite{Pradhan}.
%\begin{figure}
%%\begin{subfigure}[a]{0.47\textwidth}
%%\subcaptionbox{Computation graph with degree one node \label{fig:computation_graph_deg1}[.4\linewidth][c]{
%\centering
%\begin{tikzpicture}[every node/.style={scale=0.8}]
%\begin{scope}[node distance=2cm,>=angle 90,semithick]
%\node[cenode] (c1) {};
%\node[cnode](v1)[above of=c1] {$v_1$};
%\node(c0)[above of=v1] {};
%\draw[black] (c0) --node[very near start,left]{$e_1$} (v1);
%\draw[black] (c1) --node[very near end,left]{$e_2$} (v1);
%\node[cnode](v3)[below of=c1] {$v_3$};
%\node[cnode](v2)[below of=c1,xshift=-1cm] {$v_2$};
%\node[cnode](v4)[below of=c1,xshift=1cm] {$v_4$};
%\draw[black] (c1) --node[very near end,left]{$e_5$} (v4);
%\draw[black] (c1) --node[very near end,left]{$e_3$} (v2);
%\draw[black] (c1) --node[very near end,left]{$e_4$} (v3);
%\node(v5)[below of=v3,xshift=-0.25cm] {};
%\node(v6)[below of=v3,xshift=0.25cm] {};
%\node(v7)[below of=v4,xshift=-0.25cm] {};
%\node(v8)[below of=v4,xshift=0.25cm] {};
%\draw[black] (v3) -- (v5);
%\draw[black] (v3) -- (v6);
%\draw[black] (v4) -- (v7);
%\draw[black] (v4) -- (v8);
%\end{scope}    
%\end{tikzpicture}
%%}
%\caption{Computation graph with degree one node}
%\label{fig:computation_graph_deg1}
%\end{figure}
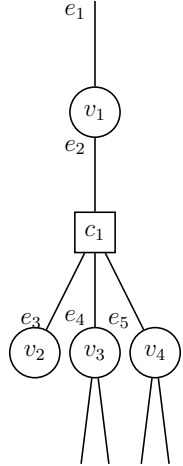
\begin{figure}
%\subcaptionbox{Protograph $G$\label{fig:Protograph_dex_illustration}}[.25\linewidth][c]{
%\centering
%\begin{tikzpicture}[every node/.style={scale=0.8}]
%\begin{scope}[node distance=2cm,>=angle 90,semithick]
%\node[cnode] (v2) {$v_1$};
%\node[cnode] (v2)[above of=v1] {$v_2$};
%\node[cnode] (v3)[above of=v2] {$v_3$};
%%\node[cnode] (v4)[above of=v3] {$v_4$};
%\node[cenode] (c1)[right of=v2] {$c_1$};
%\node[cenode] (c2)[right of=v3] {$c_2$};
%\draw[blue] (v1) --node[very near start,left]{$e_1$} (c2);
%\draw[black] (v1) --node[very near start,right]{$e_2$} (c1);
%\draw[orange] (v2) --node[very near start,below]{$e_3$} (c1);
%\draw[purple] (v3.320) --node[very near start,right]{$e_5$} (c1.120);
%\draw[red] (v3.300) --node[very near end,below]{$e_4$} (c1.140);
%\draw[green] (v3.0) --node[very near end,below]{$e_6$} (c2);
%\draw[brown] (v3.20) --node[very near start,above]{$e_7$} (c2.160);
%%\draw[black] (v4.310) --node[very near start,left]{$e_5$} (c1.140);
%%\draw[black] (v4.330) --node[near start,below]{$e_8$} (c2.140);
%%\draw[black] (v4.350) --node[near start,above]{$e_9$} (c2.120);
%\end{scope}    
%\end{tikzpicture}
%%\caption{Subgraph of protograph obtained from algorithm-\ref{deg2node_condition}.}
%}
%\subcaptionbox{Computation graph with degree one node \label{fig:computation_graph_deg1}}[.4\linewidth][c]{
\centering
\begin{tikzpicture}[every node/.style={scale=0.8}]
\begin{scope}[node distance=2cm,>=angle 90,semithick]
\node[cenode] (c1) {$c_1$};
\node[cnode](v1)[above of=c1] {$v_1$};
\node(c0)[above of=v1] {};
\draw[black] (c0) --node[very near start,left]{$e_1$} (v1);
\draw[black] (c1) --node[very near end,left]{$e_2$} (v1);
\node[cnode](v3)[below of=c1] {$v_3$};
\node[cnode](v2)[below of=c1,xshift=-1cm] {$v_2$};
\node[cnode](v4)[below of=c1,xshift=1cm] {$v_4$};
\draw[black] (c1) --node[very near end,left]{$e_5$} (v4);
\draw[black] (c1) --node[very near end,left]{$e_3$} (v2);
\draw[black] (c1) --node[very near end,left]{$e_4$} (v3);
\node(v5)[below of=v3,xshift=-0.25cm] {};
\node(v6)[below of=v3,xshift=0.25cm] {};
\node(v7)[below of=v4,xshift=-0.25cm] {};
\node(v8)[below of=v4,xshift=0.25cm] {};
\draw[black] (v3) -- (v5);
\draw[black] (v3) -- (v6);
\draw[black] (v4) -- (v7);
\draw[black] (v4) -- (v8);
\end{scope}    
\end{tikzpicture}
%}
%\subcaptionbox{Computation graph after removing degree-1 node \label{fig:introducing_deg1}}[.4\linewidth][c]{
%\centering
%\begin{tikzpicture}[every node/.style={scale=0.8}]
%\begin{scope}[node distance=2cm,>=angle 90,semithick]
%\node[cnode] (v1) {$v_1$};
%\node[cenode](c1)[below of=v1] {$c_1$};
%\node[cenode](c2)[below of=v1,xshift=2cm] {$c_2$};
%\node[cenode](c3)[below of=v1,xshift=-2cm] {$c_3$};
%\node[cnode](v3)[below of=c1,xshift=0cm] {$v_3$};
%\node[cnode](v2)[below of=c3,xshift=-0cm] {$v_2$};
%\node[cnode](v4)[below of=c2,xshift=1cm] {$v_4$};
%\node[cnode](v5)[below of=c2,xshift=-1cm] {$v_5$};
%\node (v6)[below of=v4,xshift=0.5cm] {};
%\node (v7)[below of=v5,xshift=0.5cm] {};
%\node (v8)[below of=v4,xshift=-0.5cm] {};
%\node (v9)[below of=v5,xshift=-0.5cm] {};
%\draw[black] (v1) --node[near end,left]{$e_1$} (c1);
%\draw[black] (v1) --node[near end,left]{$e_2$} (c2);
%\draw[black] (v1) --node[near end,left]{$e_3$} (c3);
%\draw[black] (c1) --node[near end,left]{$e_4$} (v3);
%\draw[black] (c3) --node[near end,left]{$e_5$} (v2);
%\draw[black] (c2) --node[near end,left]{$e_6$} (v5);
%\draw[black] (c2) --node[near end,left]{$e_7$} (v4);
%\draw[black] (v4) --node[near end,left]{$e_{11}$} (v6);
%\draw[black] (v5) --node[near end,left]{$e_{9}$} (v7);
%\draw[black] (v4) --node[near end,left]{$e_{10}$} (v8);
%\draw[black] (v5) --node[near end,left]{$e_{11}$} (v9);
%\end{scope}    
%\end{tikzpicture}
%}
\caption{Computation graph with degree one node} \label{fig:computation_graph_deg1}
%\caption{Computation Graphs}
%\label{example:double_exponential fall}
\end{figure}
Let us consider the computation graph in Fig.   \ref{fig:computation_graph_deg1} with a degree-one variable node $v_2$ and a degree-2 variable node $v_1$. In iteration $t$, we have
\begin{align}
x^t_{e_1}&=\e y^t_{e_2} \nonumber\\
\label{eq:deg1_illustration}
&=\e (1-(1-x^{t-1}_{e_3})(1-x^{t-1}_{e_4})(1-x^{t-1}_{e_5}))\nonumber\\
&\geq \e (1-(1-x^{t-1}_{e_3})) = \e^2.
\end{align} 
By using similar argument as above, it can also be shown that $y^t$ corresponding to $e_4$ and $e_5$ is less than $\e$.
This shows that the argument for double exponential fall as described in \cite{Pradhan} does not carry over directly when the protograph has degree-one variable nodes even for edges that are not directly connected to degree-$1$ variable nodes.
\subsection{Protograph LDPC code with degree-1 variable nodes}
\label{sec:double_exp_fall_standard}
We say a function $f(t)$ falls double exponentially with $t$ if $f(t)=O(\exp(-\beta 2^{\alpha t}))$ for sufficiently large $t$, with $\alpha$ and $\beta$ being positive constants. The property of block-error threshold being equal to bit-error threshold  does not require $x^t_e$ and $y^t_e$ for all $e \in E$ to fall double exponentially. It is enough to show that probability of bit-error $P_b$ corresponding to information bits falls double exponentially with iteration. Let $P_b(v)$ be the probability of bit error corresponding to a variable node $v$. We observe that $P_b(v)$ falls double exponentially if $y^t_e$ corresponding to at least one of incoming edges from a check node connected to $v$ falls double exponentially. If $x^t_e$ falls double exponentially with $t$, then the edge $e$ might help in the double exponential fall of $y^t$ corresponding to other edges of protograph. To find out the set of variable nodes for which $P_b(v)$ falls double exponentially with $t$, we need to find out set of edges for which $x_e^t$ and/or $y^t_e$ fall double exponentially with $t$.   
 
Consider a protograph $G(V \cup C ,E)$. Let $V_1 \subset V$ be the set of degree-one variable nodes and $E_1 \subset E$ be the set of edges incident on them.  Define $C_1=\{c: c$ is connected to $v\in V_1\}$. Let $G_2$ be the subgraph of $G$ induced by degree-two variable nodes. Let $E_2 \subset E$ be the set of edges of cycles in $G_2$. For a subgraph $\widehat{G}(\widehat{V} \cup \widehat{C}, \widehat{E})$ of the protograph $G$, similarly define $\widehat{V_1}, \widehat{C_1}, \widehat{G}_2$ and $\widehat{E}_2$. For $v \in  \widehat{V}$, let $E_v$ and $\widehat{E}_v$ denotes the set of edges connected to $v$ in protograph $G$ and its subgraph $\widehat{G}$, respectively. Similarly, define $E_c$ and $\widehat{E}_c$ for $c \in  \widehat{C}.$ Define $D_y =\{e: y_e^t$ falls double exponentially with $t\},$ $D_x=\{e: x_e^t$ falls double exponentially with $t\}$, $\overline{D}_x=E\setminus D_x, \overline{D}_y = E \setminus D_y, D_{xy}= D_x \cap D_y, \overline{D}_{xy}= E\setminus D_{xy}.$ 
\subsubsection{Subgraph $RED(G)$} 
 In Lemmas \ref{lemma:not_dex_fall} and \ref{lemma:not_dex_fall_hat(G)}, described in Section \ref{sec:lemma1 and 2}, it will be shown that for an edge $e \in (E_1 \cup E_2)$,  $x^t_e$ or/and $y^t_e$ does not fall double exponentially with $t$. By using the above fact, the following algorithm finds a subgraph,  denoted by $RED(G)$, such that  edges of  $RED(G)$ are not in $(\overline{D}_x \cup \overline{D}_y).$ This is done by recursively removing edges in $E_1$ and $E_2$. Observe that $E_{RED(G)}=E\setminus (\overline{D}_{x} \cup \overline D_{y})$. If $E= \overline{D}_{x} \cup \overline{D}_{y}$, then Algorithm \ref{alg:deg2node_condition} returns an empty $RED(G)$.
 \begin{algorithm}

\begin{algorithmic}
\label{alg:deg2node_condition}
\STATE
\STATE 1) $\widehat{G}=G(V\cup C,E).$
\WHILE{$\widehat{G} (\widehat{V}\cup \widehat{C}, \widehat{E})$ has variable nodes of degree-one, or the subgraph induced by degree-two variable nodes is not a tree   }
\STATE 2) $\widehat{G}_2 (\widehat{V}_2 \cup \widehat{C}_2, \widehat{E}_2)$: Subgraph induced by degree-two variable node in $\widehat{G}$. $V_2'=\{v\in \widehat{V}_2: v$  belongs to a cycle in  $\widehat{G}_2\},$ $C'_2=\{c\in \widehat{C}: c$ is connected to some $v \in V'_2 \}$
\STATE 3) $\widehat{G}=\widehat{G}-\{V'_2,C'_2\}.$ 
\STATE 4)  $V'_1=\{v\in \widehat{V}: \text{deg}(v)=1\}$, $C'_1=\{c\in \widehat{C}: \text{c is connected to some } v \in V'_1 \}$. 
\STATE 5)  $\widehat{G}=\widehat{G}-\{V'_1, C'_1\}$ (delete nodes and edges connected to them).
\ENDWHILE
\STATE 6) $RED(G)=\widehat{G}$.
 \end{algorithmic}
\end{algorithm}
\subsubsection{Double Exponential Fall}
We will now show that $x^t$ and $y^t$ corresponding to each edge $e\in   RED(G)$ fall double exponentially in the density evolution analysis of protograph $G$. For an edge $e \in E_{RED(G)}$, observe that  $E_{c_e} \subset E_{RED(G)}$.  So, from \eqref{eq:6}, it is easy to see that if $x_e^t$ falls double exponentially for all $e \in E_{RED(G)}$, then $y^t_e$ will fall double exponentially for all $e \in E_{RED(G)}$. So, it is enough to show $x^t_e$ for all $e \in RED(G)$ falls double exponentially with $t$.
\begin{theorem}
 Let  $E_{RED(G)}$  denote edges of $RED(G)$. Let
$\overline{x}^t=\max\limits_{e\in E_{RED(G)}}x^t(e)$, where $x^t(e)$ is the erasure probability along edge $e$ in the density evolution recursion of $G$. If $RED(G)$ is non-empty,  
then $E_{RED(G)} = D_{xy}.$
\label{thm:double-expon-fall}
\end{theorem}
\begin{proof}
See Section \ref{proof_thm2} for the proof.
\end{proof}
In Theorem \ref{thm:double-expon-fall}, we have shown that  $ E_{RED(G)} \subset D_{xy}.$   
%In protograph $G$, let $E_{v_e}$ and $E_{c_e}$ denotes the set of edges incident on $v_e$ and $c_e$,  respectively.  Let $\hat{E}_{v_e}$ and $\hat{E}_{c_e}$ denotes the same for any subgraph $\hat{G}$ of protograph $G$. 
If $e \in D_{xy}$, then $x^t_{e'}$ corresponding to edge $e' \in \{E_{v_e} \setminus e\}$ falls double exponentially.   
So, $\{E_{v_e}\setminus e\} \subset D_{x}$ for each $e \in E_{RED(G)}$. Now consider an edge $e \notin  E_{RED(G)}. $ We have $$y^t_e=1-\prod_{e' \in E_{c_e}\setminus e}(1-x^{t-1}_{e'}).$$ If $\{E_{c_e}\setminus e\} \subset D_x$, then $e \in D_y$. Using the above two steps,  we will find $D_y$  and $D_x$ from $E_{RED(G)}$ by using Algorithm \ref{alg:deg2node_condition_dex}. Let $r_e^t$ and $s_e^t$ denote messages on edge $e$ in $t$-th iteration from check node to variable node  and variable node to check node, respectively.
%\newpage
\begin{algorithm}

\begin{algorithmic}
\label{alg:deg2node_condition_dex}
\STATE 
%\begin{enumerate}
\STATE 1)  If $e \in E_{RED(G)}$, then initialize $r_e^0 =1,$  otherwise $r_e^0=0.$ 
%\end{enumerate}
\STATE 2) For  $e \in E$, if $\exists e'\in \{E_{v_e} \setminus e\}$ such that $r_{e'}^t=1$, then $s_e^t=1.$ 
\STATE 3) For  $e \in E$,  if $ s_{e'}^t=1$ $\forall e'\in \{E_{c_e} \setminus e\}$, then $r_{e}^{t+1}=1.$
\STATE 4) Continue 2 and 3 till $r^t_e =r^{t-1}_e.$ and $s^t_e=s^{t-1}_{e}.$
\STATE 5) $D_x=\{e\in E: s_e^t=1\}$ and $D_y=\{e\in E: r^t_e=1\}.$
\end{algorithmic}
\end{algorithm}
 Algorithm \ref{alg:deg2node_condition} and \ref{alg:deg2node_condition_dex} are illustrated through the following examples.
 \begin{example}
  Consider the rate-$1/4$ LDPC protograph $G$ in Fig. \ref{fig:wgprotograph for block-error-condition}. $G$ has a degree-one variable node $v_1$. After removing $v_1$, and the check node $c_1$ connected to $v_1$  from  $G$, we get the reduced protograph shown in  Fig. \ref{fig:ex3_without_deg1node}.  Removal of edges connected to $v_1$ and $c_1$ reduces degree of variable node $v_2$ to  one. Removal of newly introduced degree-one variable node $v_2$ and check node connected to it  introduces a cycle $v_4c_4$ formed by degree-two variable nodes. Removal of loop $v_4c_4$ from Fig. \ref{fig:ex3_reduced_graph} results in a empty $RED(G)$. So, $D_{xy}=D_y=D_x=\emptyset$.
 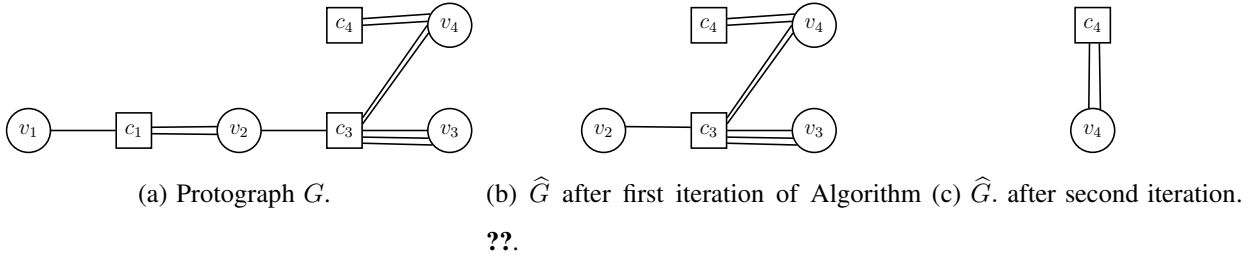
\begin{figure*}[t]
%\begin{subfigure}[a]{0.47\textwidth}
\subcaptionbox{Protograph $G$. \label{fig:wgprotograph for block-error-condition}}[.38\linewidth][c]{
\centering
\begin{tikzpicture}[every node/.style={scale=0.7}]
\begin{scope}[node distance=2cm,>=angle 90,semithick]
\node[cnode] (v1) {$v_1$};
\node[cenode] (c1)[right of=v1] {$c_1$};
\node[cnode] (v2)[right of=c1] {$v_2$};
%\node[cenode] (c2)[below of=v2] {$c_2$};
%\node[cnode] (v3)[below of=c2] {$v_3$};
\node[cenode] (c4)[right of=v2] {$c_3$};
\node[cnode] (v4)[right of=c4] {$v_3$};
\node[cnode] (v5)[above of=v4] {$v_4$};
\node[cenode] (c5)[above of=c4] {$c_4$};
\draw[black] (v1) -- (c1);
\draw[black] (c1.10) -- (v2.170);
\draw[black] (c1.350) -- (v2.190);
%\draw[black] (v2) -- (c2);
%\draw[black] (v2.260) -- (c2.100);
%\draw[black] (v2.280) -- (c2.80);
\draw[black] (v2.0) -- (c4.180);
%\draw[black] (v2.350) -- (c4.190);
%\draw[black] (c2) -- (v3);
\draw[black] (c4.0) -- (v4.180);
\draw[black] (c4.340) -- (v4.200);
\draw[black] (c4.325) -- (v4.220);
\draw[black] (c5.0) -- (v5.170);
\draw[black] (c5.20) -- (v5.150);
\draw[black] (c4.35) -- (v5.180);
\draw[black] (c4.20) -- (v5.200);
%\draw[black] (c5.340) -- (v4.160);
%\draw[black] (v2.30) -- (c5.220);
\end{scope}    
\end{tikzpicture}
%\caption{Example of a protograph which satisfy block-error threshold condition.}
%\label{fig:protograph for block-error-condition}
}
%\end{subfigure}
%\begin{subfigure}[b]{0.57\textwidth}
\subcaptionbox{$\widehat{G}$ after first iteration of Algorithm  \ref{alg:deg2node_condition}\label{fig:ex3_without_deg1node}.}[.35\linewidth][c]{
\centering
\begin{tikzpicture}[every node/.style={scale=0.7}]
\begin{scope}[node distance=2cm,>=angle 90,semithick]
\node[cnode] (v2) {$v_2$};
\node[cenode] (c4)[right of=v2] {$c_3$};
\node[cnode] (v4)[right of=c4] {$v_3$};
\node[cenode] (c5)[above of=c4] {$c_4$};
\node[cnode] (v5)[right of=c5] {$v_4$};

%\draw[black] (c4.340) -- (v4.200);
\draw[black] (c5.0) -- (v5.170);
\draw[black] (c5.20) -- (v5.150);
\draw[black] (c4.20) -- (v5.200);
%\draw[black] (c5.340) -- (v4.160);
\draw[black] (c4.35) -- (v5.180);
%\draw[black] (v2.30) -- (c5.220);
\draw[black] (v2.10) -- (c4.170);
%\draw[black] (v2.350) -- (c4.190);
\draw[black] (c4.0) -- (v4.180);
\draw[black] (c4.340) -- (v4.200);
\draw[black] (c4.325) -- (v4.220);
\end{scope}    
\end{tikzpicture}%\label{eq:8}
%\caption{Subgraph of protograph obtained from algorithm-\ref{deg2node_condition}.}
}
%\end{subfigure}
%\begin{subfigure}[b]{0.57\textwidth}
\subcaptionbox{$\widehat{G}.$ after second iteration.\label{fig:ex3_reduced_graph}}[.25\linewidth][c]{
\centering
\begin{tikzpicture}[every node/.style={scale=0.7}]
\begin{scope}[node distance=2cm,>=angle 90,semithick]
\node[cenode] (c5) {$c_4$};
%\node[cnode] (v5)[above of=v4] {$v_4$};
\node[cnode] (v5)[below of=c5] {$v_4$};
\draw[black] (c5.260) -- (v5.100);
\draw[black] (c5.290) -- (v5.70);
%\draw[black] (c5.340) -- (v5.200);
\end{scope}    
\end{tikzpicture}
%\caption{Subgraph of protograph obtained from algorithm-\ref{deg2node_condition}.}
%\label{fig:reduced_graph}%\label{eq:8}
}
%\end{subfigure}
\caption{$RED(G):$ Empty.}
\label{example:double_exponential fall}
\end{figure*}
 \end{example}
 
  \begin{example}
Consider the rate-$1/4$ protograph ($G$)  in Fig. \ref{fig:protograph for block-error-condition}. After removal of degree-one variable node $v_1$ and its neighboring check node $c_1$, we get Fig. \ref{fig:without_deg1node}. Removal of $v_1$ and $c_1$ introduces a degree-one variable node $v_2$ in Fig.  \ref{fig:without_deg1node}. After  removal of $v_2$, we get Fig.   \ref{fig:reduced_graph}, which does not have a variable node with  degree $\leq 2$. Hence, $x^t$ and $y^t$ for all edges in Fig. \ref{fig:reduced_graph} have double exponential fall property in density evolution analysis of $G$. Algorithm \ref{alg:deg2node_condition_dex} is illustrated through Fig. \ref{example:double_exponential fall_alg2}.  Algorithm \ref{alg:deg2node_condition_dex} starts by assigning $r_e^0=1$ for $e \in E_{RED(G)}$ and $r_e^0=0$ for $e \in E \setminus E_{RED(G)}$. 
%Since $\{e\in E: r_e^0=1\} \cap E_{v_3} = \{v_3c_3\}$, where $v_3c_3$ are edges between $v_3$ and $c_3$,  $s_e^0=1$ for $e \in E_{v_3}$.  Similarly, it can be shown that $s_e^0=1$ for $%\label{eq:8}e\in E_{v_4}.$ 
In Fig. \ref{example:double_exponential fall_alg2}, edges are labeled with   messages carried by them. Arrow indicates the direction of message in an edge. After end of Algorithm \ref{alg:deg2node_condition_dex}, we get $D_y= \{E_{RED(G)},v_2c_2,v_1c_1\}.$ 
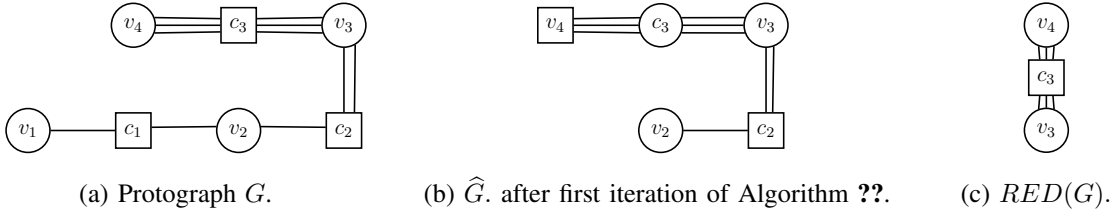
\begin{figure*}[t]
%\begin{subfigure}[a]{0.47\textwidth}
\subcaptionbox{Protograph $G$. \label{fig:protograph for block-error-condition}}[.35\linewidth][c]{
\centering
\begin{tikzpicture}[every node/.style={scale=0.7}]
\begin{scope}[node distance=2cm,>=angle 90,semithick]
\node[cnode] (v1) {$v_1$};
\node[cenode] (c1)[right of=v1] {$c_1$};
\node[cnode] (v2)[right of=c1] {$v_2$};
%\node[cenode] (c2)[below of=v2] {$c_2$};
%\node[cnode] (v3)[below of=c2] {$v_3$};
\node[cenode] (c2)[right of=v2] {$c_2$};
%\node[cnode] (v4)[right of=c4] {$v_3$};
\node[cnode] (v3)[above of=c2] {$v_3$};
\node[cenode] (c3)[left of=v3] {$c_3$};
\node[cnode] (v4)[left of=c3] {$v_4$};
\draw[black] (v1) -- (c1);%\label{eq:8}
\draw[black] (c1.10) -- (v2.170);
%\draw[black] (c1.350) -- (v2.190);
%\draw[black] (v2) -- (c2);
%\draw[black] (v2.260) -- (c2.100);
%\draw[black] (v2.280) -- (c2.80);
\draw[black] (v2.10) -- (c2.170);
%\draw[black] (v2.350) -- (c4.190);
%\draw[black] (c2) -- (v3);
%\draw[black] (c4.0) -- (v4.180);
%\draw[black] (c4.340) -- (v4.200);
\draw[black] (c3.0) -- (v3.180);
\draw[black] (c3.20) -- (v3.160);
\draw[black] (c3.-20) -- (v3.200);
\draw[black] (c2) -- (v3);
\draw[black] (c2.60) -- (v3.300);
\draw[black] (v4.0) -- (c3.180);
\draw[black] (v4.20) -- (c3.160);
\draw[black] (v4.-20) -- (c3.200);
%\draw[black] (c5.340) -- (v4.160);
%\draw[black] (v2.30) -- (c5.220);
\end{scope}    
\end{tikzpicture}
%\caption{Example of a protograph which satisfy block-error threshold condition.}
%\label{fig:protograph for block-error-condition}
}
%\end{subfigure}
%\begin{subfigure}[b]{0.57\textwidth}
\subcaptionbox{$\widehat{G}.$ after first iteration of Algorithm \ref{alg:deg2node_condition}.\label{fig:without_deg1node}}[.4\linewidth][c]{
\centering
\begin{tikzpicture}[every node/.style={scale=0.7}]
\begin{scope}[node distance=2cm,>=angle 90,semithick]
\node[cnode] (v2) {$v_2$};
\node[cenode] (c2)[right of=v2] {$c_2$};
\node[cnode] (v3)[above of=c2] {$v_3$};
\node[cnode] (c3)[left of=v3] {$c_3$};
\node[cenode] (v4)[left of=c3] {$v_4$};
\draw[black] (c3.0) -- (v3.180);
\draw[black] (c3.20) -- (v3.160);
\draw[black] (c3.-20) -- (v3.200);
\draw[black] (c2) -- (v3);
\draw[black] (c2) -- (v2);
\draw[black] (c2.70) -- (v3.290);
\draw[black] (v4.0) -- (c3.180);
\draw[black] (v4.20) -- (c3.160);
\draw[black] (v4.-20) -- (c3.200);
\end{scope}    
\end{tikzpicture}
%\caption{Subgraph of protograph obtained from algorithm-\ref{deg2node_condition}.}
}
%\end{subfigure}
%\begin{subfigure}[b]{0.57\textwidth}
\subcaptionbox{$RED(G).$ \label{fig:reduced_graph}}[.2\linewidth][c]{
\centering
\begin{tikzpicture}[every node/.style={scale=0.7}]
\begin{scope}[node distance=1cm,>=angle 90,semithick]
\node[cenode] (c3) {$c_3$};
\node[cnode] (v4)[above of=c3] {$v_4$};
\node[cnode] (v3)[below of=c3] {$v_3$};
\draw[black] (v3.90) -- (c3.270);
\draw[black] (v3.110) -- (c3.250);
\draw[black] (v3.70) -- (c3.290);
\draw[black] (c3.90) -- (v4.270);
\draw[black] (c3.110) -- (v4.250);
\draw[black] (c3.70) -- (v4.290);
\end{scope}    
\end{tikzpicture}
%\caption{Subgraph of protograph obtained from algorithm-\ref{deg2node_condition}.}
%\label{fig:reduced_graph}
}
\caption{$RED(G)$: Non empty.}
\label{example:double_exponential fall 2}
\end{figure*}
\begin{figure*}[t]
%\begin{subfigure}[a]{0.47\textwidth}
\subcaptionbox{$r_e^0$ for iteration $0$. \label{fig:re0}}[.3\linewidth][c]{
\centering
\begin{tikzpicture}[every node/.style={scale=0.5}]
\begin{scope}[node distance=2cm,>=angle 90,semithick]
\node[cnode] (v1) {$v_1$};
\node[cenode] (c1)[right of=v1] {$c_1$};
\node[cnode] (v2)[right of=c1] {$v_2$};
%\node[cenode] (c2)[below of=v2] {$c_2$};
%\node[cnode] (v3)[below of=c2] {$v_3$};
\node[cenode] (c2)[right of=v2] {$c_2$};
%\node[cnode] (v4)[right of=c4] {$v_3$};
\node[cnode] (v3)[above of=c2] {$v_3$};
\node[cenode] (c3)[left of=v3] {$c_3$};
\node[cnode] (v4)[left of=c3] {$v_4$};
\draw[middlearrow={<},black] (v1) --node[near end,above]{\large $0$} (c1);
\draw[middlearrow={>},black] (c1.10) --node[near end,above]{\large $0$}(v2.170);
%\draw[black] (c1.350) -- (v2.190);
%\draw[black] (v2) -- (c2);
%\draw[black] (v2.260) -- (c2.100);
%\draw[black] (v2.280) -- (c2.80);
\draw[middlearrow={<},black] (v2.10) --node[near end,above]{\large $0$} (c2.170);
%\draw[black] (v2.350) -- (c4.190);
%\draw[black] (c2) -- (v3);
%\draw[black] (c4.0) -- (v4.180);
%\draw[black] (c4.340) -- (v4.200);
\draw[middlearrow={>},black] (c3.0) -- (v3.180);
\draw[middlearrow={>},black] (c3.20)--node[near end,above]{\large $1$} (v3.160);
\draw[middlearrow={>},black] (c3.-20) -- (v3.200);
\draw[middlearrow={>},black] (c2) -- (v3);
\draw[middlearrow={>},black] (c2.60) -- node[near end,right]{\large $0$} (v3.300);
\draw[middlearrow={<},black] (v4.0) --node[near start,above]{\large $1$} (c3.180);
\draw[middlearrow={<},black] (v4.20) -- (c3.160);
\draw[middlearrow={<},black] (v4.-20) -- (c3.200);
%\draw[black] (c5.340) -- (v4.160);
%\draw[black] (v2.30) -- (c5.220);
\end{scope}    
\end{tikzpicture}
}
\subcaptionbox{$s_e^0$ for iteration $0$. \label{fig:se0}}[.3\linewidth][c]{
\centering
\begin{tikzpicture}[every node/.style={scale=0.5}]
\begin{scope}[node distance=2cm,>=angle 90,semithick]
\node[cnode] (v1) {$v_1$};
\node[cenode] (c1)[right of=v1] {$c_1$};
\node[cnode] (v2)[right of=c1] {$v_2$};
%\node[cenode] (c2)[below of=v2] {$c_2$};
%\node[cnode] (v3)[below of=c2] {$v_3$};
\node[cenode] (c2)[right of=v2] {$c_2$};
%\node[cnode] (v4)[right of=c4] {$v_3$};
\node[cnode] (v3)[above of=c2] {$v_3$};
\node[cenode] (c3)[left of=v3] {$c_3$};
\node[cnode] (v4)[left of=c3] {$v_4$};
\draw[middlearrow={>},black] (v1) --node[near end,above]{\large $0$} (c1);
\draw[middlearrow={<},black] (c1.10) --node[near end,above]{\large $0$} (v2.170);
%\draw[black] (c1.350) -- (v2.190);
%\draw[black] (v2) -- (c2);
%\draw[black] (v2.260) -- (c2.100);
%\draw[black] (v2.280) -- (c2.80);
\draw[middlearrow={>},black] (v2.10) --node[near end,above]{\large $0$} (c2.170);
%\draw[black] (v2.350) -- (c4.190);
%\draw[black] (c2) -- (v3);
%\draw[black] (c4.0) -- (v4.180);
%\draw[black] (c4.340) -- (v4.200);
\draw[middlearrow={<},black] (c3.0) -- (v3.180);
\draw[middlearrow={<},black] (c3.20) --node[near end,above]{\large $1$} (v3.160);
\draw[middlearrow={<},black] (c3.-20) -- (v3.200);
\draw[middlearrow={<},black] (c2) -- (v3);
\draw[middlearrow={<},black] (c2.60) -- node[near end,right]{\large $1$} (v3.300);
\draw[middlearrow={>},black] (v4.0) -- (c3.180);
\draw[middlearrow={>},black] (v4.20) -- node[near end,above]{$1$}(c3.160);
\draw[middlearrow={>},black] (v4.-20) -- (c3.200);
%\draw[black] (c5.340) -- (v4.160);
%\draw[black] (v2.30) -- (c5.220);
\end{scope}    
\end{tikzpicture}
%\caption{Subgraph of protograph obtained from algorithm-\ref{deg2node_condition}.}
}
\subcaptionbox{$r_e^1$ for iteration $1$. \label{fig:re1}}[.3\linewidth][c]{
\centering
\begin{tikzpicture}[every node/.style={scale=0.5}]
\begin{scope}[node distance=2cm,>=angle 90,semithick]
\node[cnode] (v1) {$v_1$};
\node[cenode] (c1)[right of=v1] {$c_1$};
\node[cnode] (v2)[right of=c1] {$v_2$};
%\node[cenode] (c2)[below of=v2] {$c_2$};
%\node[cnode] (v3)[below of=c2] {$v_3$};
\node[cenode] (c2)[right of=v2] {$c_2$};
%\node[cnode] (v4)[right of=c4] {$v_3$};
\node[cnode] (v3)[above of=c2] {$v_3$};
\node[cenode] (c3)[left of=v3] {$c_3$};
\node[cnode] (v4)[left of=c3] {$v_4$};
\draw[middlearrow={<},black] (v1) -- node[near end,above]{\large $0$} (c1);
\draw[middlearrow={>},black] (c1.10) -- node[near end,above]{\large $0$} (v2.170);
%\draw[black] (c1.350) -- (v2.190);
%\draw[black] (v2) -- (c2);
%\draw[black] (v2.260) -- (c2.100);
%\draw[black] (v2.280) -- (c2.80);
\draw[middlearrow={<},black] (v2.10) -- node[near end,above]{\large $1$} (c2.170);
%\draw[black] (v2.350) -- (c4.190);
%\draw[black] (c2) -- (v3);
%\draw[black] (c4.0) -- (v4.180);
%\draw[black] (c4.340) -- (v4.200);
\draw[middlearrow={>},black] (c3.0) -- (v3.180);
\draw[middlearrow={>},black] (c3.20) -- node[near end,above]{\large $1$} (v3.160);
\draw[middlearrow={>},black] (c3.-20) -- (v3.200);
\draw[middlearrow={>},black] (c2) -- (v3);
\draw[middlearrow={>},black] (c2.60) -- node[near end,right]{\large $0$} (v3.300);
\draw[middlearrow={>},black] (v4.0) -- (c3.180);
\draw[middlearrow={>},black] (v4.20) -- node[near end,above]{\large $1$} (c3.160);
\draw[middlearrow={>},black] (v4.-20) -- (c3.200);
%\draw[black] (c5.340) -- (v4.160);
%\draw[black] (v2.30) -- (c5.220);
\end{scope}    
\end{tikzpicture}
%\caption{Subgraph of protograph obtained from algorithm-\ref{deg2node_condition}.}
}
\subcaptionbox{$s_e^1$ for iteration $1$. \label{fig:se1}}[.3\linewidth][c]{
\centering
\begin{tikzpicture}[every node/.style={scale=0.5}]
\begin{scope}[node distance=2cm,>=angle 90,semithick]
\node[cnode] (v1) {$v_1$};
\node[cenode] (c1)[right of=v1] {$c_1$};
\node[cnode] (v2)[right of=c1] {$v_2$};
%\node[cenode] (c2)[below of=v2] {$c_2$};
%\node[cnode] (v3)[below of=c2] {$v_3$};
\node[cenode] (c2)[right of=v2] {$c_2$};
%\node[cnode] (v4)[right of=c4] {$v_3$};
\node[cnode] (v3)[above of=c2] {$v_3$};
\node[cenode] (c3)[left of=v3] {$c_3$};
\node[cnode] (v4)[left of=c3] {$v_4$};
\draw[middlearrow={>},black] (v1) -- node[near end,above]{\large $0$} (c1);
\draw[middlearrow={<},black] (c1.10) -- node[near end,above]{\large $1$}(v2.170);
%\draw[black] (c1.350) -- (v2.190);
%\draw[black] (v2) -- (c2);
%\draw[black] (v2.260) -- (c2.100);
%\draw[black] (v2.280) -- (c2.80);
\draw[middlearrow={>},black] (v2.10) -- node[near end,above]{\large $0$} (c2.170);
%\draw[black] (v2.350) -- (c4.190);
%\draw[black] (c2) -- (v3);
%\draw[black] (c4.0) -- (v4.180);
%\draw[black] (c4.340) -- (v4.200);
\draw[middlearrow={<},black] (c3.0) -- (v3.180);
\draw[middlearrow={<},black] (c3.20) --node[near end,above]{\large $1$} (v3.160);
\draw[middlearrow={<},black] (c3.-20) -- (v3.200);
\draw[middlearrow={<},black] (c2) -- (v3);
\draw[middlearrow={<},black] (c2.60) --node[near end,right]{\large $1$} (v3.300);
\draw[middlearrow={>},black] (v4.0) -- (c3.180);
\draw[middlearrow={>},black] (v4.20) --node[near end,above]{\large $1$} (c3.160);
\draw[middlearrow={>},black] (v4.-20) -- (c3.200);
%\draw[black] (c5.340) -- (v4.160);
%\draw[black] (v2.30) -- (c5.220);
\end{scope}    
\end{tikzpicture}
%\caption{Subgraph of protograph obtained from algorithm-\ref{deg2node_condition}.}
}
\subcaptionbox{$r_e^2$ for iteration $2$. \label{fig:re2}}[.3\linewidth][c]{
\centering
\begin{tikzpicture}[every node/.style={scale=0.5}]
\begin{scope}[node distance=2cm,>=angle 90,semithick]
\node[cnode] (v1) {$v_1$};
\node[cenode] (c1)[right of=v1] {$c_1$};
\node[cnode] (v2)[right of=c1] {$v_2$};
%\node[cenode] (c2)[below of=v2] {$c_2$};
%\node[cnode] (v3)[below of=c2] {$v_3$};
\node[cenode] (c2)[right of=v2] {$c_2$};
%\node[cnode] (v4)[right of=c4] {$v_3$};
\node[cnode] (v3)[above of=c2] {$v_3$};
\node[cenode] (c3)[left of=v3] {$c_3$};
\node[cnode] (v4)[left of=c3] {$v_4$};
\draw[middlearrow={<},black] (v1) -- node[near end,above]{\large $1$} (c1);%\label{eq:8}
\draw[middlearrow={>},black] (c1.10) -- node[near end,above]{\large $0$} (v2.170);
%\draw[black] (c1.350) -- (v2.190);
%\draw[black] (v2) -- (c2);
%\draw[black] (v2.260) -- (c2.100);
%\draw[black] (v2.280) -- (c2.80);
\draw[middlearrow={<},black] (v2.10) -- node[near end,above]{\large $1$}(c2.170);
%\draw[black] (v2.350) -- (c4.190);
%\draw[black] (c2) -- (v3);
%\draw[black] (c4.0) -- (v4.180);
%\draw[black] (c4.340) -- (v4.200);
\draw[middlearrow={>},black] (c3.0) -- (v3.180);
\draw[middlearrow={>},black] (c3.20) -- node[near end,above]{\large $1$} (v3.160);
\draw[middlearrow={>},black] (c3.-20) -- (v3.200);
\draw[middlearrow={>},black] (c2) -- (v3);
\draw[middlearrow={>},black] (c2.60) -- node[near end,right]{\large $0$} (v3.300);
\draw[middlearrow={>},black] (v4.0) -- (c3.180);
\draw[middlearrow={>},black] (v4.20) -- node[near end,above]{\large $1$} (c3.160);
\draw[middlearrow={>},black] (v4.-20) -- (c3.200);
%\draw[black] (c5.340) -- (v4.160);
%\draw[black] (v2.30) -- (c5.220);
\end{scope}    
\end{tikzpicture}
%\caption{Subgraph of protograph obtained from algorithm-\ref{deg2node_condition}.}
}

\caption{Illustration of Algorithm \ref{alg:deg2node_condition_dex}.}
\label{example:double_exponential fall_alg2}
\end{figure*}
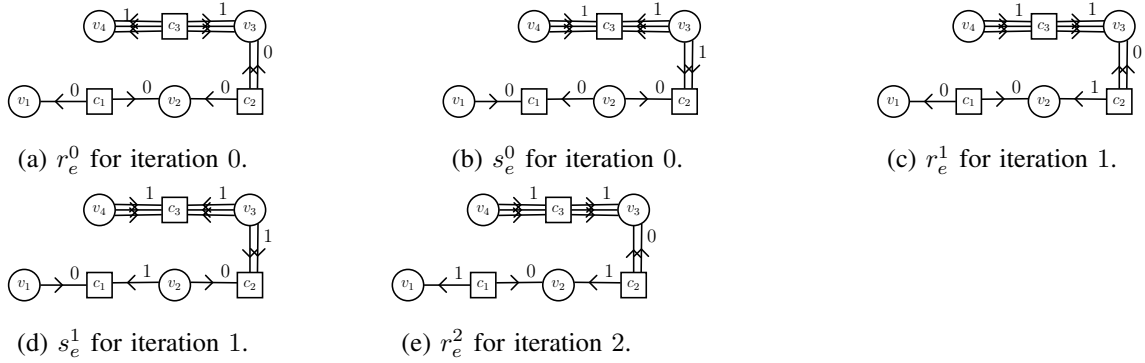

\end{example}

\subsubsection{Block-error threshold}
We now use Theorem \ref{thm:double-expon-fall} and its generalized version to a sequence of large girth liftings of a  protograph $G$ and state conditions for block-error threshold property. Let us denote the set of variable nodes of $G$ for which $P_{b}(v)$ falls double exponentially by $DEX(V)$. Let $\overline{G}$ be the code lifted from protograph $G$ with blocklength $n$ and message length $k$. Let us define $\overline{P_b}$ as $\overline{P_b}=\max_{v \in V_I}P_b(v)$, where $V_I$ is a set of variable nodes corresponding to message bits. Probability of block error can be bounded as $P_B < k\overline{P_b}$. Let $V_D$ be the variable nodes in $\overline{G}$ corresponding to $DEX(V)$. 
\begin{theorem}
\label{thm:block_error_threshold}
In the notation introduced above, if $V_I \subset V_D$, and girth of $\overline{G}$ is at least $c\log n$, then 
 $$P_B \leq k \mathcal{O}(\exp(-\beta n^{\alpha})),$$ where $\alpha, \beta, c$ are positive constants.  
\end{theorem}
\begin{proof}
 We know that $P_b(v)$ corresponding to $v \in DEX(V)$ fall double exponentially in density evolution of $G$, i.e.,$$P_b(v)=\mathcal{O}(\exp(-\beta2^{\alpha t}))$$ for $v \in DEX(V)$ with  $\alpha, \beta$ being positive constants. Since $V_I \subset V_D$, $\overline{P}_b(v)=\mathcal{O}(\exp(-\beta2^{\alpha t}))$ for $v \in V_I$. So, probability of block error of $\overline{G}$, denoted by $P_B$, can be bounded as follows: $$P_B \leq k\mathcal{O}(\exp(-\beta2^{\alpha t}))$$ for $\e \leq \e_{\text{th}}$.  Assuming $t < g/2$ and putting $t=c\log n$, we get $$P_B \leq k \mathcal{O}(\exp(-\beta n^{\alpha})).$$ 
\end{proof}
From above theorem, if $\e < \e_{\text{th}}$, we can deduce that 
$P_B \rightarrow 0$ as $n \rightarrow \infty$.
The rate-$1/4$ protograph in Fig. \ref{fig:protograph for block-error-condition} has one information bit and it satisfies the block-error threshold condition, because $P_b$ corresponding to degree-three variable node of $G$ fall double exponentially as described in Example \ref{example:double_exponential fall}. So, block-error threshold and bit-error threshold can be made equal for appropriate lifting size using Theorem \ref{thm:block_error_threshold}. Protographs can be lifted to have large girth ($\mathcal{O}(c\log n)$) by using the large girth construction in \cite{Pradhan}. Similarly, the rate-$1/4$ protograph in Fig. \ref{fig:wgprotograph for block-error-condition} has one information bit. However, for this protograph, block-error threshold  cannot be made equal to bit-error threshold by using Theorem \ref{thm:block_error_threshold}, because $P_b$ corresponding to any variable nodes of $G$ does not fall double exponentially. In the following example, we comment on the block-error threshold of codes in 5G standard \cite{5G}.

\begin{example} Consider the rate-$1/5$, $42 \times 52$ protograph and the rate-$1/3$, $46 \times 68$ protograph from the 5G standard\cite{5G}. Protographs corresponding to rate-$1/5$ and rate-$1/3$ have $10$ and $22$ information nodes, respectively. Base matrices of $RED(G)$ corresponding to rate-$1/5$ and rate-$1/3$ are given in \eqref{eq:REDG0.2} and \eqref{eq:REDG0.33}, respectively. Let $V_{RED(G)}$ denote the set of variable nodes in $RED(G)$. For both the protographs, observe that $|V_{RED(G)}|$ is greater than the number of variable nodes corresponding to information bits. Since $V_{RED(G)} \subseteq DEX(V)$, information bits in the lifted graph can be chosen in such a way that $V_I \subset V_D$. So, in 5G standard,  protographs corresponding to both rate-$1/3$ and rate-$1/5$ satisfy block-error threshold condition, which is derived in Theorem \ref{thm:block_error_threshold}.     
\end{example}
\begin{equation}
\label{eq:REDG0.2}
\left[
\begin{array}{*{14}c}
 1 & 1 & 1 & 1 & 0 & 0 & 1 & 0 & 0 & 1 & 1 & 1 & 0 & 0\\
 1 & 0 & 0 & 1 & 1 & 1 & 1 & 1 & 1 & 1 & 0 & 1 & 1 & 0\\
 1 & 1 & 0 & 1 & 1 & 0 & 0 & 0 & 1 & 0 & 1 & 0 & 1 & 1\\
 0 & 1 & 1 & 0 & 1 & 1 & 1 & 1 & 1 & 1 & 1 & 0 & 0 & 1
\end{array}
\right]\end{equation}
\begin{equation}
\label{eq:REDG0.33}
\left[
\begin{array}{*{26}c}
 1 & 1 & 1 & 1 & 0 & 1 & 1 & 0 & 0 & 1 & 1 & 1 & 1 & 1 & 0 & 1 & 1 & 0 & 1 & 1 & 1 & 1 & 1 & 1 & 0 & 0\\
 1 & 0 & 1 & 1 & 1 & 1 & 0 & 1 & 1 & 1 & 0 & 1 & 1 & 0 & 1 & 1 & 1 & 1 & 0 & 1 & 0 & 1 & 1 & 1 & 1 & 0\\
 1 & 1 & 1 & 0 & 1 & 1 & 1 & 1 & 1 & 1 & 1 & 0 & 0 & 1 & 1 & 1 & 0 & 1 & 1 & 1 & 1 & 0 & 0 & 0 & 1 & 1\\
 1 & 1 & 0 & 1 & 1 & 0 & 1 & 1 & 1 & 0 & 1 & 1 & 1 & 1 & 1 & 0 & 1 & 1 & 1 & 0 & 1 & 1 & 1 & 0 & 0 & 1
\end{array}
\right]\end{equation}
The block-error threshold condition for BIAWGN channel is same as block-error threshold condition for BEC and can be proved using a Bhattacharya parameter argument as in \cite[Theorem 3]{Pradhan} and \cite[Theorem 2]{PradhanISIT2016}. Readers interested in designing protographs with block-error threshold can skip the following sections and  move to Section \ref{sec:optim-dgldpc-prot}.  
\subsubsection{Description of Edges in $\overline{D}_x$, $\overline{D}_y$, $\overline{D}_{xy}$}
\label{sec:lemma1 and 2}
In the following two lemmas, we describe edges which are in $\overline{D}_x$, $\overline{D}_y$ and $\overline{D}_{xy}.$ 
\begin{lemma}
\label{lemma:not_dex_fall}
In the notation introduced above, the following are true: 
\begin{enumerate}
\item[1)] $E_1 \subseteq \overline{D}_x$.
\item[2)]  For  $e \in E_1$ and $e' \in E_{c_e}\setminus e$, $\{E_{c_e} \setminus e\} \subseteq \overline{D}_y$. 
\item[3)]  $E_2 \subseteq \overline{D}_{xy}$.
\item[4)] If $e \in E_2$, then $ E_{c_e} \subseteq \overline{D}_y$.
\end{enumerate}
\end{lemma}
\begin{proof}
\begin{enumerate}
\item[1)] Consider $e \in E_1$. From \eqref{eq:7}, it follows that $x^t_e=\e \quad \forall t.$ So, $E_1 \subseteq \overline{D}_x.$
\item[2)] Consider $e \in E_1.$ Observe that for each $e' \in \{ E_{c_e}\setminus e\}$, $e \in \{E_{c_{e'}} \setminus e'\}.$ From \eqref{eq:6}, it follows that  for each $e' \in \{E_{c_e}\setminus e\}$
\begin{align*}
y^t_{e'}&= 1-\prod_{\bar{e} \in E_{c_{e'}}\setminus e'}(1-x^{t-1}_{\bar{e}}) \\
&\geq x^{t-1}_{e}.
\end{align*}
Since $e \in E_1$, we know from Part 1,  $x^{t-1}_e=\e$. So, $y^t_{e'} \geq \e$ and $e'\in \overline{D}_y.$
\item[3)]  Let  ${L}=e_0e_1\cdots e_{l-1}e_0$ be a cycle in ${G_2}$. Since $L$ is a cycle,  for each $e_i$ in ${L}$, either ($v_{e_i}=v_{e_{i+1}}$ and $c_{e_{i+1}}=c_{e_{i+2}}$) or ($v_{e_i}=v_{e_{i-1}}$ and $c_{e_{i-1}}=c_{e_{i-2}}$), where addition in subscript of $e$ are modulo $l$. Next, We will prove $e_i \in \overline{D}_{x}$ if $v_{e_i}=v_{e_{i+1}}$ and $c_{e_{i+1}}=c_{e_{i+2}}$. From \eqref{eq:5}-\eqref{eq:6}, it follows that 
\begin{align*}
x^{(t+l)}_{e_i} &= \e y^{(t+l)}_{e_{i+1}}\\
&\geq \e x^{(t+l-1)}_{e_{i+2}}. 
\end{align*}
By applying \eqref{eq:5} and \eqref{eq:6} $l/2$ times alternatively, it can be shown that $x^{(t+l)}_{e_i} \geq  {\e}^{\frac{l}{2}} x^t_{e_i}$ for $1 \leq i \leq l$. Similarly, it can be proved that  $e_i \in \overline{D}_{x}$ if $v_{e_i}=v_{e_{i-1}}$ and $c_{e_{i+1}}=c_{e_{i+2}}$. So, $e_i \in \overline{D_x}$. Similarly, it can be shown that $y^{(t+l)}_{e_i} \geq {\e}^{\frac{l}{2}} (y^t_{e_i}),$ which implies $e_i \in \overline{D}_{y}$.  So, $e_i \in \overline{D}_{xy}$.
\item[4)] Consider $e \in E_2$. For each $e' \in E_{c_e}$, observe that $\{E_{c_{e'}}\setminus e'\} \cap E_2 \neq \emptyset . $ Let $\tilde{e} \in \{ E_{c_{e'}}\setminus e'\} \cap E_2.$ From \eqref{eq:5}, it follows that for each $e' \in E_{c_e}$ 
\begin{align*}
y^t_{e'}&=1-\prod_{e'' \in E_{c_{e'}}\setminus e' }(1-x^{t-1}_{e''})\\
&\geq x^{t-1}_{\tilde{e}}.
\end{align*} 
Since $\tilde{e} \in E_2$, we know from Part 3 of Lemma \ref{lemma:not_dex_fall} that $x^{t-1}_{\tilde{e}}$ does not fall double exponentially with $t$. So, $y^t_{e'}$, for $e \in E_2$ and $e' \in E_{c_e}$, does not fall double exponentially with $t.$

\end{enumerate}
\end{proof}
Define $\widehat{E}= E- \{E_1 \cup E_2\}.$ Let $\widehat{G}(\widehat{V}\cup \widehat{C})$ be the subgraph of $G$ induced by edges in $\widehat{E}$. In context of double exponential fall,  $\widehat{E}_1$ and $\widehat{E}_2$ behave in same way as $E_1$ and $E_2$, which will be shown in the following lemma.  
\begin{lemma}
\label{lemma:not_dex_fall_hat(G)}
In the notation defined above
\begin{enumerate}
\item[1)]  $\widehat{E}_1 \subseteq \overline{D}_{x}$.
\item[2)] If  $e \in \widehat{E}_1$, then $\{E_{c_e} \setminus e\} \subseteq \overline{D}_y$.
\item[3)]$\widehat{E}_2 \subseteq \overline{D}_{xy}.$

\item[4)] If $e \in \widehat{E}_2$, $E_{c_e} \subseteq \overline{D}_y.$
\end{enumerate}
\end{lemma} 
\begin{proof}
\begin{enumerate}
\item[1)] Consider $e \in \widehat{E}_1.$ We will use the fact that $x^t_e$ falls double exponentially iff $y^t_{e'}$ corresponding to at least one edge $e' \in E_{c_e}\setminus e$ falls double exponentially. We have
\begin{align*}
x^t_{e}&=\e \prod_{e' \in E_{v_e}\setminus e}y^t_{e'} \\
&=\e \prod_{e' \in \{E_{v_e}\setminus \widehat{E}_{v_e}\}}y^t_{e'}\prod_{e' \in \{\widehat{E}_{v_e}\setminus e\}}y^t_{e'}.
\end{align*}
Define $A^t_e=\prod\limits_{e' \in E_{c_e}\setminus \widehat{E}_{c_e}}y^t_{e'}$. From Lemma \ref{lemma:not_dex_fall}, it can be deduced that $A^t_e$ does not fall double exponentially with $t$. Since $e$ is incident on a degree-one variable node of $\widehat{G}$,  $\widehat{E}_{c_e} \setminus e = \emptyset.$ So, $x^t_{e}=\e A^t_e$ and it does not fall double exponentially with $t$.
\item[2)] Similar to the proof of Part 2 of Lemma \ref{lemma:not_dex_fall}.
\item[3)]  Let  $\widehat{L}=e_0e_1\cdots e_le_0$ be a cycle in $\widehat{G_2}$. Since $\widehat{L}$ is a cycle,  for each $e_i$ in $\widehat{L}$, either ($v_{e_i}=v_{e_{i+1}}$ and $c_{e_{i+1}}=c_{e_{i+2}}$) or ($v_{e_i}=v_{e_{i-1}}$ and $c_{e_{i-1}}=c_{e_{i-2}}$), where addition in subscript of $e$ are modulo $l$. Next, we will prove $e_i \in \overline{D}_{x}$ if $v_{e_i}=v_{e_{i+1}}$ and $c_{e_{i+1}}=c_{e_{i+2}}$. From \eqref{eq:5} and \eqref{eq:6}, we have 
\begin{align*}
x^{t+l}_{e_i}&= \e \prod_{e \in E_{v_{e_i}}\setminus e_i}y^{t+l}_{e}\\
& = \e \prod_{e \in \{E_{v_{e_i}}\setminus \widehat{E}_{v_{e_i}}\}}y^{t+l}_{e}\prod_{e \in \widehat{E}_{v_{e_i}}\setminus e_i}y^{t+l}_{e}\\
&=A_e^{t+l}y^{t+l}_{e_{i+1}}\\
&\geq A_e^{t+l}x^{t+l-1}_{e_{i+2}},
\end{align*}  
where $A^{t+l}_{e_i}=\e\prod\limits_{e \in \{E_{v_{e_i}}\setminus \widehat{E}_{v_{e_i}}\}}y^{t+l}_{e}$ and $e_{i+1}=\widehat{E}_{v_{e_i}}\setminus e_i.$ From Lemma \ref{lemma:not_dex_fall}, it can be deduced that $A^{t+l}_{e_i}$ does not fall double exponentially. By applying above $l/2$ times we can show that $$x^{t+l}_{e_i}\geq A_{e_1}^{t}x^{t}_{e_{i}}.$$ Since $A^t_{e_i}$ does not fall double exponentially, $x^{t}_{e_i}$ does not fall double exponentially with $t$. Similarly, it can be proved that  $e_i \in \overline{D}_{x}$ if $v_{e_i}=v_{e_{i-1}}$ and $c_{e_{i+1}}=c_{e_{i+2}}$. So, $e_i \in \overline{D_x}$. Similarly, it can be proved that $e_i \in \overline{D}_y$. So, $e_i \in \overline{D}_{xy}$.   
\item [4)] Similar to the proof of Part 4 of Lemma \ref{lemma:not_dex_fall}.
\end{enumerate}
\end{proof}

\subsubsection{Proof of Theorem \ref{thm:double-expon-fall}}
\label{proof_thm2}
\begin{proof}
The proof follows the proof of \cite[Theorem 1]{Pradhan} very
closely. We will briefly sketch the proof here.
We will use the following inequality. For any $x \in [0,1]$ and a positive integer $d$, 
\begin{align}
\label{eq:ineq}
(d-1)x \geq 1-(1-x)^{d-1}
\end{align}
First observe that $E_{RED(G)} \subseteq E$. If $G$ contains degree-one variable nodes or cycles in the subgraph induced by degree-two variable nodes, then $E_{RED(G)}\subset E$. Let $|v_2|$ be the number of degree two variable nodes in $RED(G)$. Let us consider $e_{v_i,m} \in E_{RED(G)}$. For $l \in \{0,1,2, \cdots |v_2|\}$, we will show by recursion that
\begin{align}
\label{eq:recursion}
 x_{v_i,m}^{t+l} \leq C_l\left(\overline{x}^t\right)^{a(l,e_{v_i,m})}%\leq \widehat{C_l}\left(\overline{x}^t\right)^{\widehat{a}(l,e_{v_i,m})},
 \end{align} 
 where  $C_l$ is a constants. We have $C_0 =1$ and 
 $a(0,e_{v_i,m})=1.$ so, \eqref{eq:recursion} is true for $l=0$. In standard protograph, single parity check codes and repetition codes are used as component code at check nodes and variable nodes, respectively. So, for standard protograph, \eqref{eq:6} and \eqref{eq:7} becomes \eqref{eq:spc} and \eqref{eq:repetition}, respectively.
 \begin{align} 
 \label{eq:spc}
 y^{t+1}_{c_j,n}&=1-\displaystyle \prod_{k \in [d_{c_j}]\setminus n}(1-x^t_{c_j,k}),\\
 \label{eq:repetition}
 x^{t+1}_{v_i,m}&= \e \prod_{k \in [d_{v_i}]\setminus m}y^{t+1}_{v_i,k}.
\end{align} 
  where $[d_{c_j}]=\{1,2,\cdots,d_{c_j}\},$ and $[d_{v_i}]=\{1,2,\cdots,d_{v_i}\}.$ Let ${b}(l,e_{c_j,n})=\displaystyle\min_{k\in [d_{c_j}]\setminus n}{a}(l,e_{c_j,k}).$ Next, we will prove \eqref{eq:recursion} for arbitrary $l$. Using \eqref{eq:spc} and $\overline{x}_t \leq 1$, we get 
  \begin{align*}
  y^{t+l+1}_{c_j,n}&= 1-\prod_{k \in [d_{c_j}]\setminus n}(1-x^{t+l}_{c_j,k}),\\
  &\stackrel{\text{(a)}}{\leq} 1-\left(1-C_l\left(\overline{x}^t\right)^{b(l,e_{c_j,n})}\right)^{d_{c_j}-1},\\
 % &\stackrel{\text{b}}{\leq} 1-\left(1-C_l\left(\overline{x}^t\right)^{\widehat{b}(l,e_{c_j,n})}\right)^{d_{c_j}-1},\\
  &\stackrel{\text{(b)}}{\leq} (d_{c_j}-1)C_l \left(\overline{x}^t\right)^{{b}(l,e_{c_j,n})}.
  \end{align*}
  Inequality $(a)$  follows from \eqref{eq:recursion}. Since $\overline{x}^t\rightarrow 0$ for $\e \leq \e_{\text{th}}$, we have $C_l\left(\overline{x}^t\right)^{b(l,e_{c_j,n})} < 1$ for large $t$. So, inequality $(b)$ follows from \eqref{eq:ineq}. Consider $e_{v_i,m} \in E_{RED(G)}$. We get

\begin{align*} 
 x^{t+l+1}_{v_i,m}&=\e \prod_{n \in [d_{v_i}]\setminus m}y^{t+l+1}_{v_i,n},\\
 &\leq \e \left((d_{\max}-1)C_l\right)^{(d_{v_i}-1)} \prod_{n \in [d_{v_i}]\setminus m}\left(\overline{x}^t\right)^{{b}(l,e_{c_j,n})},\\
 &\stackrel{\text{a}}{\leq} \e \left((d_{\max}-1)C_l\right)^{(d_{v_i}-1)} \prod_{n \in [d_{v_i}(RED(G))]\setminus m}\left(\overline{x}^t\right)^{{b}(l,e_{(c_j,n)})},\\
 &\leq C_{l+1} \left(\overline{x}^t\right)^ {a(l+1,e_{(v_i,m)})}, \\
 %&\leq C_{l+1} \left(\overline{x}^t\right)^ {\widehat{a}(l+1,e_{(v_i,m)})}, 
\end{align*}
where $d_{\max}$ is the maximum check node degree $(d_{\max} \geq 2)$, $d_{v_i}$ and $d_{v_i}(RED(G))$ are degree of node $v_i$ in $G$ and $RED(G)$, respectively, $C_{l+1}= \displaystyle \max_{1\leq i \leq |V|} \e \left((d_{\max}-1)C_l\right)^{(d_{v_i}-1)}$  is a positive constant and we set
\begin{align}
a(l+1,e_{v_i,m})=\begin{cases}
1, \text{ if }\displaystyle \sum_{k \in [d_{v_i}(RED(G))]\setminus m}b(l,e_{(v_i,k)})=1,\\
2, \text{ if }\displaystyle\sum_{k \in [d_{v_i}(RED(G))\setminus m] }b(l,e_{(v_i,k)})\ge2.  
\end{cases}
\end{align}
Inequality (a) is true, because $d_{v_i} \geq d_{v_i}(RED(G))$.
 We claim that $a(v_2+1,e_{v_i,m})=2$. This can be proved by contradiction. For details of the proof, we refer readers to \cite[Theorem 1]{Pradhan}.
 So, we have shown that
\begin{align*}
\overline{x}^{t+v_2+1} \leq A (\overline{x}^t)^2,
\end{align*}
where $A=C_{t+v_2+1}$ is a constant and $\overline{x}^t \leq 1$ for $t>R$. By applying the above repeatedly, we can show that 
\begin{align}
\label{eq:11}
\overline{x}^{R+i(2v_2+1)}\leq A^{-1}(A\overline{x}^R)^{2^i},
\end{align}
 for every positive integer $i$, which implies $E_{RED(G)}\subseteq D_{xy}.$ Next, we will prove  $D_{xy} \subseteq E_{RED(G)}.$
In Algorithm \ref{alg:deg2node_condition}, we remove edges in the set $\overline{D}_x\cup \overline{D}_y \cup \overline{D}_{xy}$ from $G$ to obtain graph $RED(G)$, i.e, $E_{RED(G)}=E\setminus \left(\overline{D}_x\cup \overline{D}_y \cup \overline{D}_{xy} \right)$. In Lemmas \ref{lemma:not_dex_fall} and \ref{lemma:not_dex_fall_hat(G)}, it has been shown that $D_{xy} \cap \left(\overline{D}_x\cup \overline{D}_y \cup \overline{D}_{xy} \right) =\emptyset.$ So, $D_{xy} \subseteq E_{RED(G)}$
  and the proof of the theorem is complete.
\end{proof}

\subsection{Extension to DGLDPC protograph}
For an edge $e$ in a DGLDPC protograph $G$, let $h_{ce}$ and $h_{ve}$ denote the extrinsic message erasure probabilities at the check node and variable node, respectively. Note that $h_{ce}$ and $h_{ve}$ are polynomials in multiple variables denoting erasure probabilities of edges in $E_{c_e}\setminus e$ and $E_{v_e} \setminus e$, respectively (see \eqref{eq:5}-\eqref{eq:7}). Let $d_{ce}$ and $d_{ve}$ denote the least sum degree of terms in $h_{ce}$ and $h_{ve}$, respectively. Let $D_x, D_y, \overline{D}_x, \overline{D}_y, D_{xy},$ and $\overline{D}_{xy}$ denote the same quantities as in Section \ref{sec:double_exp_fall_standard}.  As shown in Lemma \ref{lemma:not_dex_fall}, in a standard protograph, edges from degree-$1$ variable nodes do not contribute to double exponential fall, because $d_{ve}$ corresponding to them is zero. Unlike standard protograph, $d_{ve}$ corresponding to an edge $e$ in DGLDPC protograph is not determined by degree of its variable node.  In a DGLDPC protograph, let $E_1^v=\{e: d_{ve}=0\}$ and $E_1^c=\{e: d_{ce}=0\}$, i.e $E_1^v$ and $E_1^c$ are the set of edges with a constant term in their corresponding $h_{ce}$ and $h_{ve}$, respectively. Define $E_1 = E_1^v \cup E_1^c$. Let $G_2$ be the subgraph induced by edges in $\{e: d_{ce}=1 \text{ or } d_{ve}=1\}.$ A loop $L=\{e_0e_1 \cdots e_{2l-1}e_0\}$ in $G_2$ is said to be non-$DEX$ (does not have double exponential fall property) if it satisfies one of the following conditions.
\begin{enumerate}
\item[1)] Degree-one term of $h_{ve_i}$ and $h_{ce_{i+1}}$, for $0\leq i \leq 2l-1$, are $a^v_{e_{i+1}}y_{e_{i+1}}$ and $a^c_{i+2}x_{e_{i+2}}$, respectively.
\item[2)] Degree-one term of $h_{ce_i}$ and $h_{ve_{i+1}}$, for $0\leq i \leq 2l-1$, are $a^c_{e_{i+1}}x_{e_{i+1}}$ and $a^v_{i+2}y_{e_{i+2}}$, respectively.
\end{enumerate}
In the above, $a^v_{e_i}, a^v_{e_{i+1}}, a^c_{e_{i}}$ and $a^c_{e_{i+1}}$ are constants, and addition in subscript of $e$ is modulo $2l$. If loop $L$ satisfies  condition-1 above, then define $E_2^{Lv}=\{e_0, e_2, \cdots, e_{2l-2}\}$ and $E_2^{Lc}=\{e_1, e_3, \cdots, e_{2l-1}\}$, else define $E_2^{Lc}=\{e_0, e_2, \cdots, e_{2l-2}\}$ and $E_2^{Lv}=\{e_1, e_3, \cdots, e_{2l-1}\}$. Define $E_2^L=E_2^{Lc} \cup E_2^{Lv}$ and $E_2=\underset{{L \in \mathcal{L}}}{\bigcup}E_2^L$, where $\mathcal{L}$ is the set of non-$DEX$ cycles in $G_2$. In the following lemma, we describe edges which are in $\overline{D}_x$ and $\overline{D}_{y}$. 
\begin{lemma}
\label{lemma:not_dex_fall_gen}
In the notation introduced above, the following are true: 
\begin{enumerate}
\item[1)] $E_1^v \subseteq \overline{D}_x$ and $E_1^c \subseteq \overline{D}_y$.
\item[2)] For an edge $e$ with $d_{ce}=1$, if one of the degree-1 term of $h_{ce}$ is $a_{e'}x_{e'}$ for $e' \in E_1^v$, then $e' \in \overline{D}_y$. Similarly, for an edge $e$ with $d_{ve}=1$, if one of the degree-1 term of $h_{ve}$ is $a_{e'}y_{e'}$ for $e' \in E_1^c$, then $e' \in \overline{D}_x$. 
\item[3)]  $E_2^{Lc} \subseteq \overline{D}_y$ and $E_2^{Lv} \subseteq \overline{D}_x$.
\item[4)] For an edge $e$ with $d_{ce}=1$, if the degree one of the degree-1 term of $h_{ce}$ is $a_{e'}x_{e'}$ for $e' \in E_2^{Lv}$, then $ e \in \overline{D}_y$. Similarly, for an edge $e \in G_2$,  if one of the degree-1 term of $h_{ve}$ is $a_{e'}y_{e'}$ for $e' \in E_2^{Lc}$, then $e \in \overline{D}_x$. 
\end{enumerate}
\end{lemma}
\begin{proof}
\begin{enumerate}
\item[1)] Consider $e \in E_1^v$. From definition of $E_1^v$, we know that $d_{ve}=0$. Let $a_e$ be the degree-zero term of $h_{ve}$. From \eqref{eq:7}, it follows that $x^{t}_e \geq a_e$. So, $E_1^v \subseteq \overline{D}_x$. Similarly, it can be shown that $E_1^c \subseteq \overline{D}_y$. 
\item[2)] From \eqref{eq:6}, it follows that $y^{t+1}_e \leq a_{e'}x_{e'}^t $. So, $e' \in \overline{D}_y$. Other statement can be proved similarly.
\item[3)] Consider $e_i \in E_2^{Lv}$. From definition of $E_2^{Lv}$ and \eqref{eq:6}-\eqref{eq:7}, it follows that
\begin{align*}
x^{t+l}_{e_i} &\geq a_{e_{i+1}}y^{t+l}_{e_{i+1}}\\
&\geq a_{e_{i+1}}a_{e_{i+2}}x^{t+l-1}_{e_{i+2}}.
\end{align*}
Repeating the above $l$ times, we get $$x^{t+l}_{e_i} \geq \left(\displaystyle \prod_{i=1}^{l}a^c_{2i}a^v_{2i+1}\right)x^{t}_{e_i}.$$ 
So, $E_2^{Lv} \subseteq \overline{D}_x$. Similarly, it can be shown that  $E_2^{Lc} \subseteq \overline{D}_y$.
\item[4)] Consider an edge $e \in G_2$. Since the degree-1 term of $h_{ce}$ is $a_{e'}x_{e'}$, from \eqref{eq:7}, it follows that $y_{e'}^{t+1}\geq a_{e'}x_{e'}^{t}$. Since $e' \in E_2^{Lv}$, $e \in \overline{D}_y$. Other statement can be proved similarly.
\end{enumerate}
\end{proof}
Algorithm \ref{alg:deg2node_condition} is modified to recursively remove $E_1$ and $E_2$ to find $RED(G)$ for a DGLDPC protograph $G$. Another modification is as follows: after removing an edge, the message erasure probability of all edges are updated by replacing the message corresponding to the removed edge by $1.$ For example, consider $x^{t+1}_{e_1}=\epsilon y^{t+1}_{e_2}$. After removal of edge $e_2$, it becomes $x^{t+1}_{e_1}=\epsilon$. Define $\widehat{E}= E- \{E_1 \cup E_2\}.$ Let $\widehat{G}(\widehat{V}\cup \widehat{C})$ be the subgraph of $G$ induced by edges in $\widehat{E}$. In context of double exponential fall,  $\widehat{E}_1$ and $\widehat{E}_2$ behave in same way as $E_1$ and $E_2$, proof of which is similar to Lemma \ref{lemma:not_dex_fall_hat(G)}.   
Algorithm \ref{alg:deg2node_condition_dex} is extended with no significant modification. After obtaining $D_x$ from Algorithm \ref{alg:deg2node_condition_dex}, Theorem \ref{thm:block_error_threshold} can be applied directly to find whether a DGLDPC protograph satisfies the block-error threshold condition. The following lemma plays a role in relating minimum distance of component codes to double exponential fall property.

\begin{lemma}(\cite{mct}[Theorem 3.79])
 \label{exitfunction_min_distance}
 For a linear code with minimum distance $d$, let
$f_i(\epsilon)$ be the probability that the extrinsic output of the
MAP decoder over BEC$(\e)$ is an erasure for the $i$-th bit.
 Then, $f_i(\e)$ is a polynomial in $\e$ such that the coefficient of $\e^i$
 is nonzero only for $i\ge d-1$. 
\end{lemma}
\begin{proof}
 For proof, see \cite{mct}[Theorem 3.79].
\end{proof}
\begin{remark}
In a DGLDPC protograph $G(V \cup C, E)$, let $\overline{E_1^v}=\{e\in E:$ minimum distance of component code at $v_e$ is $1\}$ and $\overline{E_1^c}=\{e\in E:$ minimum distance of component code at $c_e$ is $1\}$. Let $\overline{E_2}$ be the set of edges in loop formed by nodes having component code with minimum distance $2$. From Lemma \ref{exitfunction_min_distance}, it follows that $\overline{E_1^v} \subseteq E_1^v, \overline{E_1^c} \subseteq E_1^c,$ and $\overline{E_2} \subseteq E_2$. $\overline{E_1^v}, \overline{E_1^c},$ and  $\overline{E_2}$ are removed recursively to obtain $\overline{RED(G)}$. Observe that $\overline{RED(G)}$ is a subgraph of $RED(G)$.    
\end{remark}

In this section, we have derived necessary and sufficient condition for block-error threshold. In the next section, we will use the block-error threshold condition to design protographs with block-error threshold close to capacity. 
\section{Optimized DGLDPC Protographs}
\label{sec:optim-dgldpc-prot}
In this section, we design capacity-approaching protographs  with block-error threshold by using the condition derived in Section \ref{sec:prop-block-error}. Let $G$ be a protograph of size $|V| \times |C|$, where $V$ and $C$ denote the set of variable and check nodes. We  divide variable nodes in $V$ into two sets - standard variable nodes denoted as $V_s$, and generalized variable nodes denoted as $V_g$. $C_s$ and $C_g$ are similar notations for check nodes. We use repetition code and SPC code at standard variable nodes and check nodes, respectively. At a generalized node $v$, we choose a $(d_{v},k_{v})$  linear code as component code. To design a rate-$r$ code, we choose component codes at generalized nodes in such a way that $r=1-\frac{\sum_{i=1}^{|C|}(d_{c_i}-k_{c_i})}{\sum_{i=1}^{|V|}k_{v_i}}$. At standard variable node $v$ and check node $c$, we have $k_v=1$ and $d_c-k_c=1.$ We maximize the  block-error threshold of protograph over the connections of protograph, degree of standard nodes, and label of edges connected to generalized nodes by using differential evolution \cite{DE}.
\subsection{Differential Evolution} 
Different steps of differential evolution are elaborated as follows. The details of optimizing labels of edges at generalized nodes is skipped for brevity.
\begin{itemize}
\item[1)]Initialization is done as follows
\begin{itemize}
\item[$\bullet$] Start with $|C||V|$ base matrices $B_{k,0}$, $0\leq k \leq |C||V|$, each of size $|C|\times|V|$. To restrict the search space, entries of base matrices are chosen randomly from the set $\{0,1,\cdots,8\}$.  Enforce  variable and check node degree constraint at generalized nodes, i.e. $\sum_{i=1}^{|C|} B_{k,0}(i,j) = d_{v_j} $, for $v_j\in V_g$, $\sum_{j=1}^{|V|}B_{k,0}(i,j) = d_{c_i}$,  $c_i \in C_g$.
\end{itemize}
If $B_{k,0}$ does not satisfy block-error threshold condition derived in Theorem \ref{thm:block_error_threshold}, add an edge between degree-1 or degree-2 standard variable node and standard check node, chosen randomly. Continue adding such edges till the block-error threshold condition is satisfied.
\item[2)]Mutation: Protographs of  generation $N$ $(N=0, 1,\cdots)$ are 				 interpolated as follows.
\begin{align}
 M_{k,N}=[B_{r_1,N}+0.5(B_{r_2,N}-B_{r_3,N})],
\end{align}
where $r_1$, $r_2$, $r_3$ are randomly-chosen distinct values, and $[x]$ denotes the absolute value of $x$ rounded to
the nearest integer.
 \item[3)] Crossover: A candidate protograph $B'_{k,N}$ is chosen as
  follows. The $(i,j)$-th entry of $B'_{k,N}$ is set as the $(i,j)$-th
  entry of $M_{k,N}$ with probability $p_c$, or as the $(i,j)$-th
  entry of $B_{k,N}$ with probability $1-p_c$. We use $p_c=0.88$ in
  our optimization runs. If $B_{k,N}(i,j)=B'_{k,N}(i,j)$,   labels of  the edges corresponding to $B'_{k,N}(i,j)$ are copied to labels of edges corresponding to $B_{k,N}(i,j),$ otherwise   edges corresponding to $B'_{k,N}$ are labeled randomly  without assigning same label  to two edges connected to same node.   
 \item[4)]Selection:  If the bit-error threshold of $B_{k,N}$ is greater than that of
  $B'_{k,N}$ and it satisfies block-error threshold condition in Theorem  \ref{thm:block_error_threshold}, set $B_{k,N+1}=B_{k,N}$; else, set $B_{k,N+1}=B'_{k,N}$.
  \item[5)]Termination: Steps 2--4 are run for several generations (we
  run up to $N=6000$) and the protograph that gives the best block-error threshold is chosen as the optimized protograph. 
\end{itemize} 

We compute thresholds of protographs for the BEC by using the density evolution described in Section \ref{sec:density_evolution_BEC}. We compute thresholds of protograph for AWGN channel using the EXIT function method described in \cite{Sharon}.
%\subsection{Lifting Protographs}
%\label{subsec:liftingprotograph}
%A $|v| \times |c|$ protograph $G$ is lifted to a LDPC code of length $n$ as follows:
%\begin{itemize}
%\item[1.] Let $m$ denote the maximum element of the corresponding $|V|\times|C|$ base matrix $B$. Replace  $B(i,j)$ with a random $m \times m$ binary matrix $M$ with row and column sum equal to $B(i,j)$. If $B(i,j)=0$,  replace by an $m \times m$ zero matrix.  Denote the new $|V|m \times |C|m$ binary matrix as $B'$.
%\item[2.] We assume $m|V|$ divides $n$. Replace each nonzero entry of $B'$ with a $\frac{n}{|V|m} \times \frac{n}{|V|m}$, randomly generated, circular shift permutation matrix.  
%\end{itemize}
%There can be more sophisticated liftings. However, we use the above simple method in all simulations and do not attempt to optimize the lifting further.
\subsection{Optimized protographs for BEC}

For BEC, optimized LDPC protographs (base matrices) of rate $1/10$ and $1/8$ with block-error thresholds within $0.01$ of  capacity are given in \eqref{th=0.894} and \eqref{th=0.866}, respectively, in the Appendix. It is observed that optimized protographs have significant fraction of degree-one variable nodes. Thresholds of LDPC protographs with degree-one nodes, LDPC protographs without degree-one nodes, GLDPC protographs with degree-one nodes and AR4A protograph \cite{5174517} for BEC have been compared in Table \ref{tab:optprot_BEC}. We see that optimized protographs have better thresholds when degree-one nodes are allowed in optimization. For example, an optimized, rate-$1/8$,  protograph with degree-one bit nodes in \eqref{th=0.866} has threshold 0.866 over BEC, while optimized, rate-$1/8$ protograph without degree-one nodes has a threshold 0.85. In optimization of DGLDPC protograph,  $(7,4)-$Hamming code and its dual are used as component codes at generalized variable nodes and generalized check nodes, respectively. For example, $8 \times 10$, rate-$1/10$ DGLDPC protograph in Table \ref{tab:optprot_BEC}, has two generalized check nodes and two generalized variable nodes. From simulation, it is observed that increasing number of generalized node does not improve the block-error threshold. From Table \ref{tab:optprot_BEC},  it is also observed that use of a generalized component code does not improve the threshold. From simulation, it is observed that use of other linear codes, such as Hadamard code, as component code does not improve the block-error threshold. However, generalized nodes are useful in designing smaller protographs with block-error threshold reasonably close to capacity. For example, an optimized $8 \times 10$, rate-$1/8$ DGLDPC protograph has block-error threshold 0.86 which is quite close to $0.866$ achieved with a $21 \times 24$ LDPC protograph. 

Optimized protographs in Table \ref{tab:optprot_BEC}  are lifted to codes of blocklength 5000 using cyclic progressive edge growth described in \cite{cPEG} and their BER/FER are simulated using the standard message-passing decoder. The plots are shown in  Fig \ref{fig:simulation}(a). For comparison,  AR3A/AR4A \cite{1523619} protographs  are lifted to the same blocklength of $5000$ using the method in \cite{cPEG} and their BER/FER are plotted in Fig. \ref{simulation_BEC}. We see that the BER and FER of optimized codes are better than that of AR4A codes of same rate.

\begin{table}[htp]
  \centering
  \begin{tabular}{|c|c|c|c|c|}
    \hline
    {\bf{Rate}}&\makecell{\bf Size of \\ \bf Protograph }&\makecell{\bf Types of \\ \bf Protograph }&\makecell{\bf DE \\ \bf Threshold}&\makecell{\bf Block \\ \bf Threshold}\\
    \hline
    \multirow{2}{*}{1/10}&$27\times 30$ in \eqref{th=0.894} & \makecell{LDPC with \\ degree-1} & 0.894 & Yes\\[1.2ex]
                                   \cline{2-5}
                                   &\makecell{$10\times 11$ \\ in \cite[Fig. 12] {1523619}}&{AR4JA}&0.868&No\\
                                    \cline{2-5}
                                   &$17\times 23$&\makecell{GLDPC with \\ degree-1}&0.892&Yes\\
                                   \cline{2-5}
                                    &$27\times 30$&\makecell{LDPC w/o \\ degree-1}&0.877&Yes\\
                                     \cline{2-5}
                                   &$12\times 14$&\makecell{DGLDPC w/o \\ degree-1}&0.893&Yes\\               
    \hline
    \multirow{2}{*}{1/8}&$21\times 24$ in \eqref{th=0.866} & \makecell{LDPC with \\ degree-1} &0.866&Yes\\
                                   \cline{2-5}
                                   &\makecell{$8\times 9$ \\in \cite[Fig. 11]{1523619}}&{AR4JA}&0.846&No\\
                                    \cline{2-5}
                                   &$13\times 19$&\makecell{GLDPC with \\ degree-1}&0.866&Yes\\
                                   \cline{2-5}
                                   &$14\times 16$  \cite{Pradhan}&\makecell{LDPC w/o \\ degree-1}&0.85&Yes\\ 
                                    \cline{2-5}
                                   &$8\times 10$&\makecell{DGLDPC w/o \\ degree-1}&0.86&Yes\\                        
    \hline
%    \multirow{3}{*}{1/6} &$20\times 24$ in \eqref{th:0.823}&\makecell{LDPC with \\ degree-1}&0.823&Yes\\
%                                   \cline{2-5}
%                                   &\makecell{$6\times 7$  \\ in \cite[Fig. 10]{1523619}}&{AR4JA}&0.81&No\\  
%                                    \cline{2-5}
%                                   &$10\times 14$&\makecell{GLDPC with \\ degree-1}&0.822&Yes\\                             
%                                   \cline{2-5}
%                                   &$10\times 12$  \cite{Pradhan}&\makecell{LDPC w/o \\ degree-1}&0.812&Yes\\
%                                    \cline{2-5}
%                                   &$9\times 12$&\makecell{DGLDPC w/o \\ degree-1}&0.824&Yes\\
%    \hline
  \end{tabular}
  \caption{Optimized protographs and thresholds for BEC.}
  \label{tab:optprot_BEC}
\end{table}

\subsection{Optimized Protographs For AWGN}
% For the BIAWGN channel, optimized LDPC protographs (base matrices) of rate $1/3$ and $1/4$ with block-error thresholds within $0.3$dB of  capacity are given in \eqref{th:-0.330} and \eqref{th:-0.630}, respectively, in the Appendix. Thresholds over AWGN channel have been compared in Table \ref{tab:optprot_AWGN}. 
Observations similar to the BEC case hold for AWGN channel as well. Fig. \ref{fig:simulation}(b) compares FER of optimized  rate-$1/3$ and rate-$1/5$  codes with protographs of same rate from 5G standard\cite{5G},  PBRL family \cite{pbrl}, and  AR4A family. All protographs are lifted to codes having blocklength around 64000. Parity check matrix corresponding to optimized protographs, protographs in 5G standard, and AR4A protographs are obtained by cyclic progressive edge growth described in \cite{cPEG}.  Optimized protographs in this work have better block-error threshold than protographs  of same rate in $5G$ standard by $0.1$dB. We also observe that allowing multiple edges between same pair of nodes enables to design protograph of smaller size with comparable threshold. For example, the $28 \times 41$, rate -$1/3$ protograph in this work which allows multiple edges between nodes has block-error threshold of $-0.405$dB, whereas the $46 \times 68$, rate-$1/3$ protograph in 5G standard which does not have multiple edges between nodes has a block-error threshold of $-0.225$dB. Although block-error threshold condition is derived assuming infinite blocklength, FER performance of optimized protographs are better than their corresponding codes in 5G, PBRL, and AR4A protographs when blocklength is finite as shown in Fig \ref{fig:simulation}(b). For example, at FER=$10^{-3}$ and blocklength 64000, the rate-$1/5$ protograph in \eqref{th:-0.8} has a gap of $0.463$dB to capacity, whereas the rate-$1/5$ protograph in 5G standard has a gap of $0.513$dB to capacity.

\begin{table}
  \centering
  %\resizebox{\columnwidth}{!}{
  \begin{tabular}{|c|c|c|c|c|}
     \hline
    {\bf{Rate}}&\makecell{\bf Size of \\ \bf Protograph }&\makecell{\bf Types of \\ \bf Protograph }&\makecell{\bf DE \\ \bf Threshold}&\makecell{\bf Block \\ \bf Threshold}\\
    \hline
    \multirow{1}{*}{1/5}&$34\times 42$ in \eqref{th:-0.8} &\makecell{LDPC with \\ degree-1}&-0.834&Yes\\
    							 \cline{2-5}
                                   & $42 \times 52$ &{5G}&-0.714&Yes\\
                                    \cline{2-5}
                                   &$4\times 5$ in \cite[Fig.8]{1523619}&{AR4A}&-0.522&No\\
    \hline
    \multirow{1}{*}{1/3}&$28\times 41$ in \eqref{th:-0.330}  &\makecell{LDPC with \\ degree-1}&-0.405&Yes\\
    							 \cline{2-5}
                                   &$17 \times 25$ &{PBRL}&-0.150&Yes\\
                                    \cline{2-5}
                                    \cline{2-5}
                                   &$46 \times 68$ &{5G}&-0.225&Yes\\
                                    \cline{2-5}
                                   &$3\times 4$ in \cite[Fig.7]{1523619}&{AR4A}&-0.130&No\\
    \hline
  \end{tabular}
 % }
  \caption{Optimized protographs and thresholds for AWGN.}
  \label{tab:optprot_AWGN}
\end{table}

 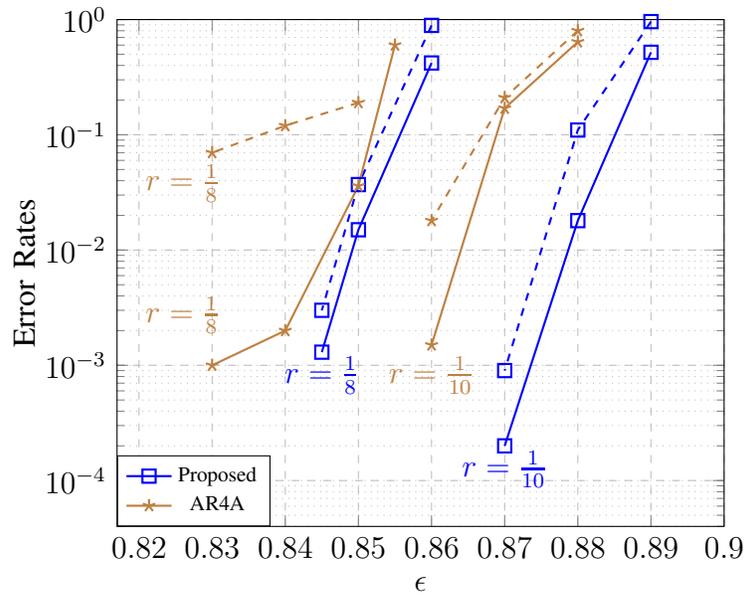
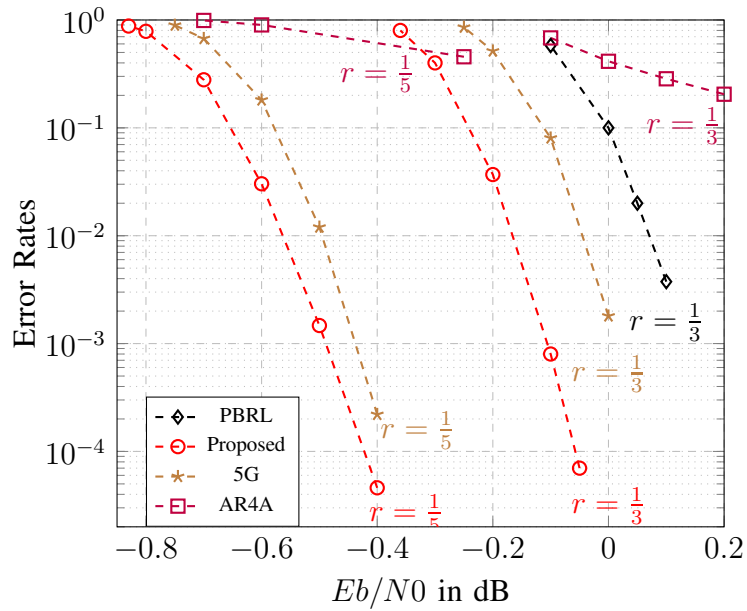
\begin{figure}

\begin{subfigure}[htb]{1\textwidth}
   \centering
  \begin{tikzpicture}
\begin{semilogyaxis}[width=3.8in,xmin=0.817,xmax=0.9,ymin=4e-5,ymax=1,grid=both,minor grid style=dotted,major grid style=dashed,xlabel=$\epsilon$,ylabel=Error Rates,legend style={font=\fontsize{8}{9}\selectfont,at={(axis cs:0.817,4e-5)},anchor=south west}]
\addplot [color=blue,solid,mark=square,mark size=2.5,mark options={solid},thick] 
coordinates {
    (0.89,5.2e-1)
    (0.88,1.8e-2)
    (0.87,2e-4)
}node[below,pos=0.02] {} node[below] at (axis cs:0.87,2.3e-4){$r=\frac{1}{10}$};
\addplot [color=blue,dashed,mark=square,mark size=2.5,mark options={solid},thick] 
coordinates {
    (0.89,9.6e-1)
    (0.88,1.1e-1)
    (0.87,9e-4)
};
\addplot [color=blue,solid,mark=square,mark size=2.5,mark options={solid},thick]
coordinates {
    (0.86,4.2e-1)
    (0.85,1.5e-2)
    (0.845,1.3e-3)
}node[below,pos=0.02] {}node[below] at (axis cs:0.845,1.5e-3){$r=\frac{1}{8}$};  
\addplot [color=blue,dashed,mark=square,mark size=2.5,mark options={solid},thick]
coordinates {
    (0.86,8.9e-1)
    (0.85,3.7e-2)
    (0.845,3e-3)
};
%\addplot [color=red,solid,mark=square,mark size=2.5,mark options={solid},thick]
%coordinates{
%    (0.815,2e-1)
%    (0.81,5.3e-2)
%    (0.80,2.4e-3)
%}node[below,pos=0.02] {}node[below] at (axis cs:0.805,2.4e-3){$r=\frac{1}{6}$};;
%\addplot [color=red,dashed,mark=square,mark size=2.5,mark options={solid},thick]
%coordinates{
%    (0.815,5e-1)
%    (0.81,1.4e-1)
%    (0.80,4e-3)
%};

\addplot [color=brown,solid,mark=star,mark size=2.5,mark options={solid},thick]
coordinates{
    (0.88,6.4e-1)
    (0.87,1.7e-1)
    (0.86,1.5e-3)
}node[below,pos=0.02] {}node[below]at (axis cs:0.86,1.5e-3) {$r=\frac{1}{10}$};
\addplot [color=brown,dashed,mark=star,mark size=2.5,mark options={solid},thick]
coordinates{
   (0.88,8e-1)
    (0.87,2.1e-1)
    (0.86,1.8e-2)
};
\addplot [color=brown,solid,mark=star,mark size=2.5,mark options={solid},thick]
coordinates{
	(0.855,6e-1)
    (0.85,3.6e-2)
    (0.84,2e-3)
    (0.83,1e-3)
}node[above,pos=1,left] {}node[below] at (axis cs:0.826,5e-3) {$r=\frac{1}{8}$};
\addplot [color=brown,dashed,mark=star,mark size=2.5,mark options={solid},thick]
coordinates{
   (0.85,1.9e-1)
    (0.84,1.2e-1)
    (0.83,7e-2)
}node[below,pos=0.1,left] {}node[below] at (axis cs:0.826,7e-2) {$r=\frac{1}{8}$};
%\addplot [color=green,solid,mark=star,mark size=2.5,mark options={solid},thick]
%coordinates{
%    (0.81,2.9e-1)
%    (0.80,3.8e-2)
%    (0.79,1e-3)
%}node[below] at (axis cs:0.79,1e-3){$r=\frac{1}{6}$}node[below,pos=0.1] {};
%\addplot [color=green,dashed,mark=star,mark size=2.5,mark options={solid},thick]
%coordinates{
%    (0.81,8e-1)
%    (0.80,5.3e-1)
%    (0.79,3.9e-1)
%}node[below,pos=0.1,left] {}node[below] at (axis cs:0.79,3.9e-1){$r=\frac{1}{6}$};
%\legend{Proposed r=1/10,,Proposed r=1/8,,Proposed r=1/6,,AR4A r=1/10,,AR4A r=1/8,,AR4A r=1/6}
\legend{Proposed ,, ,,,,AR4A , ,}
\end{semilogyaxis}
\end{tikzpicture}
\caption{Performance of codes over BEC, length=5000.}
\label{simulation_BEC}
\end{subfigure}
\begin{subfigure}[htb]{1\textwidth}
   \centering
\begin{tikzpicture}
\begin{semilogyaxis}[width=3.8in,xmin=-0.85,xmax=0.2,ymin=2e-5,ymax=1,grid=both,minor grid style=dotted,major grid style=dashed,xlabel=$Eb/N0$ in dB,ylabel=Error Rates,legend style={font=\fontsize{8}{9}\selectfont,at={(axis cs:-0.8,2e-5)},anchor=south west}]
%\addplot [color=red,solid,mark=o,mark size=2.5,mark options={solid},thick] 
%coordinates {
%    (-0.4,6e-6)
%    (-0.5,1.74e-5)
%    (-0.6,4.51e-4)
%    (-0.7,5.99e-3)
%    (-0.8,4.14e-2)
%    (-0.83,5.26e-2)
%}node[below] at (axis cs:-0.75,5e-4) {$r=\frac{1}{5}$};
\addplot [color=red,dashed,mark=o,mark size=2.5,mark options={solid},thick] 
coordinates {
    (-0.4,4.6e-5)
    (-0.5,1.47e-3)
    (-0.6,3.03e-2)
    (-0.7,2.79e-1)
    (-0.8,7.85e-1)
    (-0.83,8.79e-1)
}node[below] at (axis cs:-0.35,5.1e-5) {$r=\frac{1}{5}$};

%\addplot [color=brown,solid,mark=star,mark size=2.5,mark options={solid},thick] 
%coordinates {
%    (-0.4,9e-7)
%    (-0.5,8.47e-4)
%    (-0.6,1.49e-2)
%    (-0.7,7.22e-2)
%    (-0.75,1.61e-1)
%    
%}node[below] at (axis cs:-0.45,3e-3) {$r=\frac{1}{5}$};
\addplot [color=brown,dashed,mark=star,mark size=2.5,mark options={solid},thick] 
coordinates {
	(-0.4,2.2e-4)
    (-0.5,1.20e-2)
    (-0.6,1.81e-1)
    (-0.7,6.73e-1)
    (-0.75,8.94e-1)
}node[below] at (axis cs:-0.33,3e-4) {$r=\frac{1}{5}$};
%\addplot [color=purple,solid,mark=square,mark size=2.5,mark options={solid},thick] 
%coordinates {
%    %(-0.45,3e-6)
%    (-0.25,3e-3)
%    (-0.6,1.91e-1)
%    (-0.7,2.08e-1)
%    
%}node[below] at (axis cs:-0.25,4e-3) {$r=\frac{1}{5}$};
\addplot [color=purple,dashed,mark=square,mark size=2.5,mark options={solid},thick] 
coordinates {
	(-0.25,4.56e-1)
    (-0.6,9e-1)
    (-0.7,9.90e-1)
}node[below] at (axis cs:-0.4,6e-1) {$r=\frac{1}{5}$};
%\addplot [color=black,solid,mark=diamond,mark size=2.5,mark options={solid},thick] 
%coordinates {
%	(0.1,5.05e-4)
%    (0.05,2.69e-3)
%    (0,1.26e-2)
%    (-0.1,7.88e-2)
%    
%}node[below] at (axis cs:0.1,3e-3) {$r=\frac{1}{3}$};
\addplot [color=black,dashed,mark=diamond,mark size=2.5,mark options={solid},thick] 
coordinates {
	(0.1,3.75e-3)
    (0.05,2e-2)
    (0,1e-1)
    (-0.1,5.8e-1)
    
}node[below] at (axis cs:0.1,3e-3) {$r=\frac{1}{3}$};
%\addplot [color=red,solid,mark=o,mark size=2.5,mark options={solid},thick] 
%coordinates {
%    (-0.1,7.38e-6)
%    (-0.2,4.92e-4)
%    (-0.3,8.69e-3)
%    (-0.36,2.45e-2)
%    
%}node[below] at (axis cs:0,6e-5) {$r=\frac{1}{3}$};
\addplot [color=red,dashed,mark=o,mark size=2.5,mark options={solid},thick] 
coordinates {
    (-0.05,7e-5)
    (-0.1,8e-4)
    (-0.2,3.69e-2)
    (-0.3,4e-1)
    (-0.36,8e-1)
}node[below] at (axis cs:0,6e-5) {$r=\frac{1}{3}$};
%\addplot [color=brown,solid,mark=star,mark size=2.5,mark options={solid},thick] 
%coordinates {
%    (0,2.29e-4)
%    (-0.1,9.82e-3)
%    (-0.2,6.60e-2)
%    (-0.25,1.13e-1)
%    
%}node[below] at (axis cs:0,3e-4) {$r=\frac{1}{3}$};
\addplot [color=brown,dashed,mark=star,mark size=2.5,mark options={solid},thick] 
coordinates {
	(0,1.80e-3)
    (-0.1,8.04e-2)
    (-0.2,5.15e-1)
    (-0.25,8.56e-1)
    
}node[below] at (axis cs:0,1e-3) {$r=\frac{1}{3}$};
%\addplot [color=purple,solid,mark=square,mark size=2.5,mark options={solid},thick] 
%coordinates {
%    (0.2,2.35e-4)
%    (0.1,8.76e-4)
%    (0,1.33e-2)
%    (-0.1,8.00e-2)
%    
%}node[below] at (axis cs:0.13,2e-1) {$r=\frac{1}{3}$};
\addplot [color=purple,dashed,mark=square,mark size=2.5,mark options={solid},thick] 
coordinates {
	(0.2,2.05e-1)
    (0.1,2.85e-1)
    (0,4.14e-1)
    (-0.1,6.82e-1)
    
}node[below] at (axis cs:0.13,2e-1) {$r=\frac{1}{3}$};
\legend{,,,PBRL,Proposed,5G,AR4A}
\end{semilogyaxis}
\end{tikzpicture}
\caption{Performance of codes over AWGN, length=64000.}
\label{simulation_AWGN}
\end{subfigure}
\caption{BER:Solid, FER:Dashed.}
\label{fig:simulation}
\end{figure}

\section{Conclusion}
In summary, we designed low-rate codes with block-error threshold close to capacity. From simulation, we  observe that optimized codes have better FER performance than comparable protographs of same rate. This work provides a theoretical basis for LDPC codes in 5G standard.
\bibliographystyle{IEEEtran}
\bibliography{IEEEabrv,reference}
\begin{appendices}
  \section{Optimized Codes} 
  Non zero entries of optimized base matrices corresponding to different rates are  given below. Non zero entries of the $i$-th row of a base matrix is listed next to $i:$. Superscript denotes the element at that location. If superscript is not mentioned, then non-zero element at that location is $1$. Variable nodes corresponding to first two column of rate-$1/5$ and rate-$1/3$, respectively, in \eqref{th:-0.8} and $\eqref{th:-0.330}$ are punctured.
 %\setlength{\arraycolsep}{1.3pt}
 %\newpage
  %\begin{strip}
  \begin{align}
  \label{th:-0.330}
 & 1: 1, 3, 10, 13, 20, & 2:& 1, 2, 3, 7, 10, & 3:& 4, 5, 6, 7, 9,  & 4:& 2^2, 5, 6, 14, 25.\nonumber\\
&\quad  30.&& 33. && 38. & 5:& 2, 4, 5, 7, 25,\nonumber\\
& 6: 3, 4, 6, 11, 12,& 7:& 1, 3, 10, 12, 26.& 8:& 1, 2, 4, 7, 40.&&32.\nonumber\\
&\quad 27.& 9:& 1, 2, 29. & 10:& 2, 3, 9, 12, 18,  & 11:& 1, 3, 10, 11, 13,  \nonumber\\
& 12: 1, 2, 4, 36. & 13:& 2, 3, 4, 8, 21, && 23, 28.&& 18, 22, 24.\nonumber\\
& 14: 1, 2, 3, 9, 21.&& 24. & 15:& 3, 4, 7, 8, 15,  & 16:& 1, 3, 4, 17, 22, \nonumber\\
& 17: 1, 3, 4, 9, 35. & 18:& 1, 2, 6, 31. && 17. && 27.\nonumber\\ 
&19: 1, 2, 4, 8, 37. & 20:& 3^2, 5, 6, 9, 14, & 21:& 1, 3, 6, 8, 13, & 22:& 2, 4, 5, 10, 15, \nonumber\\
&23: 1, 2, 3, 12, 17, && 15, 16. && 16. && 19, 26.\nonumber\\
&\quad 41. & 24:& 1, 2, 8, 39. & 25:& 2, 3, 5, 7, 14,  & 26:& 1, 2, 5, 6, 11,\nonumber\\
& 27: 1, 3, 7, 11, 16,  & 28:& 1, 2, 4, 5, 20. && 23. && 34.\nonumber\\
& \quad 19, 28.
 \end{align}
\begin{align}
\label{th:-0.8}
%& 1: 1, 2, 5, 16. & 2:& 1, 2, 3, 5, 13. & 3:& 2^2, 3, 5, 10, 14, & 4:& 1, 15, 19, 24.\nonumber\\ & 5: 1^2, 6.& 6:& 1, 4, 5^2, 6, 7,&&15. & 7:& 2, 4^2, 5, 6, 15^2, \nonumber\\
% & 8: 1, 3, 5, 8, 23.&&12, 14, 19, 24.& 9:& 2, 6, 12, 21, 24. && 17, 23, 24.\nonumber\\
% & 10: 3, 6, 7, 21. & 11:& 1, 3, 4, 5, 9. & 12:& 2, 5, 15, 23. & 13:& 1, 2^2, 3^2, 4^2, 5^2,\nonumber\\
% & 14: 1^2, 5, 18.& 15:& 1, 2, 5, 15.& 16:& 2, 4, 5, 14, 21.&&6^2, 11, 14, 23^2.\nonumber\\
% & 17: 1^2, 2, 20. & 18:& 1, 3, 4, 5. & 19:& 1, 5, 6, 22. & 20:& 3, 4, 5, 15, 25.\nonumber\\
% & 21: 1, 5, 6, 26.
& 1: 1, 6, 8, 10, 22. & 2:& 2, 3, 7, 13, 39. & 3:& 3, 4, 9, 17, 35. & 4:& 1, 2, 11.\nonumber\\
& 5: 2, 4, 15, 18, 19. & 6:& 1, 2, 3^2, 26. & 7:& 1, 3, 21. & 8:& 2, 4, 12, 40.\nonumber\\
& 9: 1, 5, 9, 32. & 10:& 2, 8, 20, 21, 22. & 11:& 1, 4, 6, 30. & 12:& 2, 5, 10, 16, 18, 20.\nonumber\\
& 13: 1^2, 3, 27. & 14:& 3, 4, 5, 12, 14, 34. & 15:& 1, 2, 3, 42. & 16:& 2, 3, 5, 7.\nonumber\\
& 17: 2, 4, 10, 16, 38. & 18:& 6, 7, 8, 11, 41. & 19:& 1, 2, 8, 14, 19. & 20:& 2, 4, 9, 16.\nonumber\\
& 21: 1, 2, 3, 31. & 22:& 1, 2, 4, 29. & 23:& 2, 3, 10, 13. & 24:& 1, 2, 12, 24.\nonumber\\
& 25: 1, 2, 15, 25. & 26:& 3, 5, 8, 23, 28. & 27:& 1, 6, 7, 14. & 28:& 1, 2, 11, 33.\nonumber\\
& 29: 2, 3, 4, 17. & 30:& 6, 9^2, 15, 24. & 31:& 5, 6^2, 7, 13, 37. & 32:& 1, 3, 14, 23.\nonumber\\
& 33: 1, 4, 5, 12, 13. & 34:& 1, 4, 11, 17, 36. 
\end{align}
\begin{align}
  \label{th=0.894}
 & 1: 3, 21, 29.  &2:& 10, 21, 24.  & 3:& 11, 12, 14, 21. & 4:& 6, 12, 21, 30. & 5:& 18, 21, 29.\nonumber\\   & 6: 8, 20, 21^2, 23. &7:& 7, 9, 11. & 8:& 3, 8, 21, 28, 30.	& 9:& 16^2,18. & 10:& 10, 21, 25, 26.\nonumber\\
& 11: 4, 11, 25.   & 12:& 11, 25, 29.  & 13:& 5, 11, 25.  & 14:& 2, 16, 21, 24. & 15:& 21, 24, 27.\nonumber\\
& 16: 11, 25^2.  & 17:& 1, 3, 14, 18, 21^2.  &18:& 8, 9, 17, 21.  & 19:& 7, 14, 21, 25.   & 20:& 11, 24, 25, 30.\nonumber\\
 & 21: 11, 15, 22. & 22:& 3, 9, 13, 30. & 23:& 19, 21^2.  & 24:& 8, 11, 20.  & 25:& 3^2, 7, 13, 14, \nonumber\\ & 26: 8, 11, 21. & 27:& 15, 25, 29.&&&&&& 19, 24, 25, 30.
  \end{align}
  \begin{align}
  \label{th=0.866}
 & 1: 2, 5, 9, 12, 13,  & 2:& 13, 15, 18, 23. & 3:& 6, 10, 12, 13. & 4:& 8, 10, 13, 17.  & 5:& 7, 13, 15.\nonumber\\& \quad 17, 23. & 6:& 13, 15, 17. & 7:& 6, 7, 11, 15, & 8:& 4, 7, 13, 18. & 9:& 7, 14, 15, 17.\nonumber\\
 &10:7^2,12,21.&11:&3,7^2.&& 18. &12:&7,13,15,19.&13:& 7,13,18^2,24.\nonumber\\
 & 14: 7^2, 17, 20. & 15:& 2, 7, 12. &16:& 7^2, 9, 13, 22. &17:& 9,13^2,15^2,& 18:& 6, 7, 9, 13.\nonumber\\& 19: 7^2, 12. & 20:& 7, 11, 18,22^2,& 21:& 1, 7, 12, 15.&&16,17.\nonumber\\&&&24.
 \end{align}
 \end{appendices}

\end{document}